\documentclass[acmtocl]{acmtrans2m}

\usepackage{amsmath}
\usepackage{amssymb}
\usepackage{array}
\usepackage{color}
\usepackage{graphicx}
\usepackage{mathrsfs}
\usepackage{multicol}
\usepackage{multirow}
\usepackage{xspace}

\newcommand{\dlv}{{\sc dlv}\xspace}
\newcommand{\gidl}{{\sc gidl}\xspace}

\newcommand{\msid}{{\sc MiniSAT(ID)}\xspace}

\newcommand{\N}{\mathbb{N}}									
\newcommand{\res}[2]{{#1}|_{#2}}								
\newcommand{\lub}{\mathrm{lub}}								
\newcommand{\glb}{\mathrm{glb}}								

\newcommand{\size}[1]{\lvert {#1} \rvert}					

\newcommand{\pol}{{\bf P}\xspace}
\newcommand{\np}{{\bf NP}\xspace}
\newcommand{\conp}{{\bf coNP}\xspace}

\newcommand{\ddd}{\overline{d}}

\newcommand{\ttt}{\overline{t}}
\newcommand{\xxx}{\overline{x}}
\newcommand{\yyy}{\overline{y}}
\newcommand{\zzz}{\overline{z}}

\newcommand{\tval}{\text{\bf v}}								
\newcommand{\true}{\text{\bf t}}								
\newcommand{\false}{\text{\bf f}}							
\newcommand{\unkn}{\text{\bf u}}								
\newcommand{\inco}{\text{\bf i}}								

\newcommand{\voc}{\Sigma}										

\newcommand{\id}[1]{\ensuremath{\text{\it {#1}}\,}}	

\newcommand{\lir}{\Rightarrow}								
\newcommand{\ler}{\Leftrightarrow}							
\newcommand{\rul}{\leftarrow}									

\newcommand{\agg}{\text{\sc f}}								
\newcommand{\card}{\text{\sc card}}							
\newcommand{\asum}{\text{\sc sum}}							
\newcommand{\aprod}{\text{\sc prod}}						
\newcommand{\amin}{\text{\sc min}}							
\newcommand{\amax}{\text{\sc max}}							

\newcommand{\open}[1]{\mathrm{Open}({#1})}				
\newcommand{\defi}[1]{\mathrm{Def}({#1})}					
\newcommand{\wfm}[2]{\mathrm{wfm}_{#1}({#2})}			
\newcommand{\comp}[1]{\mathrm{Comp}({#1})}				

\newcommand{\leqt}{\leq_t}										
\newcommand{\geqt}{\geq_t}										
\newcommand{\leqp}{\leq_p}										
\newcommand{\geqp}{\geq_p}										

\newcommand{\app}[1]{\tilde{{#1}}}							
\newcommand{\lapp}{\bot^{\leqp}}								
\newcommand{\mapp}{\top^{\leqp}}								
\newcommand{\ct}[1]{#1_{\rm ct}}								
\newcommand{\cf}[1]{#1_{\rm cf}}								
\newcommand{\tf}[1]{{\rm tf}(#1)}								
\newcommand{\prp}{O}												
\newcommand{\iprop}{\text{\sc inco}}						
\newcommand{\rfi}[1]{\mathscr{O}^{{#1}}}					
\newcommand{\trsl}[1]{\mathscr{L}_{{#1}}}					
\newcommand{\infp}[1]{\mathscr{I}^{{#1}}}					
\newcommand{\infps}[1]{\mathscr{I}({#1})}					
\newcommand{\infs}[1]{\text{\textup{INF}}({#1})}		

\newcommand{\symvoc}{\Upsilon}									
\newcommand{\prps}{S}											
\newcommand{\sinfp}[1]{\mathscr{I}_s^{{#1}}}				

\newcommand{\foagg}{\textup{FO(AGG)}\xspace}				
\newcommand{\foid}{\textup{FO(ID)}\xspace}				
\newcommand{\mcl}{\mathcal{L}}								

\newcommand{\dom}[1]{\text{\textit{dom}}({#1})}

\newcommand{\ignore}[1]{}

\newcommand{\ico}[1]{\Gamma_{#1}}

\newcommand{\domf}{\text{\textit{dom}}}

\acmYear{11}
\acmMonth{June}

\title{Constraint Propagation for First-Order Logic and Inductive Definitions}
\author{JOHAN WITTOCX, MARC DENECKER, and MAURICE BRUYNOOGHE \\ Katholieke Universiteit Leuven}

\begin{abstract} 
Constraint propagation is one of the basic forms of inference in many logic-based reasoning systems. In this paper, we investigate constraint propagation for first-order logic (FO), a suitable language to express a wide variety of constraints. We present an algorithm with polynomial-time data complexity for constraint propagation in the context of an FO theory and a finite structure. We show that constraint propagation in this manner can be represented by a datalog program and that the algorithm can be executed symbolically, i.e., independently of a structure. Next, we extend the algorithm to FO(ID), the extension of FO with inductive definitions. Finally, we discuss several applications.

\end{abstract}

\newtheorem{theorem}{Theorem}[section]

\newtheorem{proposition}[theorem]{Proposition}
\newtheorem{lemma}[theorem]{Lemma}
\newdef{definition}[theorem]{Definition}
\newdef{remark}[theorem]{Remark}
\newdef{example}[theorem]{Example}
\newdef{algotext}[theorem]{Algorithm}

\begin{document}

\maketitle

\section{Introduction}

An interesting trend in declarative problem solving is the growing overlap between research in constraint programming (CP), propositional satisfiability (SAT) and certain subareas of knowledge representation and reasoning (KRR). In CP, we witness the evolution towards more expressive, logic-like languages. The same evolution is also witnessed in the SAT community, where there is a growing interest in SAT modulo theories (SMT), i.e., solving satisfiability problems for a much richer language than propositional logic.  In KRR, attention is shifting from deduction as main reasoning task towards other forms of inference. These evolutions are leading to an apparent {\em congruence} between the problems and the languages studied in these areas. In CP, one searches for {\em assignments} to {\em variables} that satisfy certain {\em constraints}~\cite{Apt03}. While originally, variables ranged over finite atomic domains, in recent rich solver-independent CP-languages like {\sc essence}~\cite{constraints/FrischHJHM08} and Zinc~\cite{constraints/MarriottNRSBW08}, variables also range over complex types such as arrays and sets. There is a close match with the logical inference problem of \emph{finite model generation}, in which {\em structures (i.e., models)} are searched interpreting a logical {\em vocabulary} consisting of constant, function and predicate symbols that satisfy a set of {\em logical propositions}. Not coincidentally, recently new approaches for search and optimization emerged that use expressive logics with origins in the area of knowledge representation and solve such problems through model generation inference. The approach was pioneered in Answer Set Programming (ASP)~\cite{marek99stable,Baral03}; now also systems based on (extensions of) first order logic (FO) are available.  The best solvers of this kind embrace technologies from (mainly) SAT and offer superior modelling environments that already now prove particularly well-suited in {\em knowledge-intense} search or optimization problems of bounded size.  The growing overlap between CP, SAT and KRR is further witnessed by recent efforts to include CP techniques in ASP \cite{MellarkodGZ08,iclp/GebserOS09}, by the successful participation of the ASP solver {\sc clasp} in the SAT competition \cite{url:sat2009competition}, and by the participation of the constraint logic programming system B-prolog in the ASP competition \cite{lpnmr/DeneckerVBGT09}.

In this paper, we push the convergence between CP and KRR a step further by studying constraint propagation for classical first-order logic (FO). To this end, we first define constraint propagation for FO. Informally, for a given FO theory $T$ and a finite partial structure $\app{I}$, constraint propagation boils down to computing facts that are certainly true or certainly false in every total structure that satisfies $T$ and that ``completes'' $\app{I}$. To illustrate this definition, consider a database application allowing university students to compose their curriculum by selecting certain didactic modules and courses. Assume that amongst others, the following integrity constraints are imposed on the selections: 
\begin{align}
	& \forall x \forall y\ (\id{MutExcl}(x,y) \lir \neg(\id{Selected}(x) \land \id{Selected}(y))), \label{form:cc}	\\
	& \exists m\ (Module(m) \land \id{Selected}(m)),                      \label{form:om}	\\
	& \forall c\ (\id{Course}(c) \land \exists m\ (\id{Module}(m) \land \id{Selected}(m) \land \id{In}(c,m)) \lir \id{Selected}(c). \label{form:ac}
\end{align}
The first constraint states that mutually exclusive components cannot be selected both, the second one expresses that at least one module should be taken and the third one ensures that all courses of a selected module are selected. Consider a situation where there are, amongst others, two mutually exclusive courses $c_1$ and $c_2$, that $c_1$ belongs to a certain module $m_1$, and that at some point in the application, the student has selected $m_1$ and is still undecided about the other courses and modules. That is, an incomplete database or \emph{partial structure} is given. One can check that in every total selection that completes this partial selection and satisfies the constraints, $c_1$ will be selected, $c_2$ will not be selected, and no module containing $c_2$ will be selected. Constraint propagation for FO aims to derive these facts.

Given a theory $T$ and a partial structure $\app{I}$, computing all the models of  $T$ that complete $\app{I}$, and making facts true (respectively false) that are true (respectively false) in all these models yields the most precise results. However, it is in general too expensive to perform constraint propagation in this way.
The constraint propagation algorithm we present in this paper is less precise, but, for a fixed theory $T$, runs in polynomial time in the domain size of $\app{I}$. The algorithm consists of two steps. First, $T$ is rewritten in linear time to an equivalent theory $T'$ such that for each constraint in $T'$, there exists a precise polynomial-time propagator. In the second step, these propagators are successively applied, yielding polynomial-time propagation for $T'$, and hence for $T$.

Besides its polynomial-time data complexity, our algorithm has two other benefits. First, the propagation can be represented by a set of (negation-free) \emph{rules} under a least model semantics. Such sets of rules occur frequently in logic-based formalisms. Examples are Prolog, Datalog, Stable Logic Programming \cite{marek99stable,Niemela99}, FO extended with inductive definitions \cite{tocl/DeneckerT08}, and production rule systems \cite{phd/Forgy79}. As a consequence, many of the theoretical and practical results obtained for these formalisms can be applied to study properties of our method, as well as to implement it efficiently.  Secondly, it is possible to execute the propagation symbolically, i.e., independently of the given partial structure. Symbolic propagation is useful in, e.g., applications where the partial structure is subject to frequent changes. 

As can be deduced from many logics developed in KRR and CP, (inductive) definitions and aggregates are two concepts that are crucial to model many real-life applications. Yet in general, these concepts cannot be expressed in FO. To broaden the applicability of our propagation algorithm, we extend it to FO(ID), the extension of FO with inductive definitions\cite{tocl/DeneckerT08}. An extension to aggregates is discussed in Appendix~\ref{sec:aggr}. 

In the last part of this paper, we sketch several applications of our propagation algorithm, namely model generation, preprocessing for grounding, configuration, approximate query answering in incomplete databases and conformant planning. 

This paper is an extended and improved presentation of \cite{WittocxMD2008:KR}. It describes (part of) the theoretical foundation for applications presented in \cite{jair/WittocxMD10,iclp/WittocxVD09,jelia/VlaeminckWVDB10,ppdp/VlaeminckVD09}. A less densely written version of this paper is part of the PhD thesis of the first author \cite{phd/Wittocx10}.

\section{Preliminaries}

We assume the reader is familiar with classical first-order logic (FO). In this section, we introduce the notations and conventions used throughout this paper and we recall definitions and results about three- and four-valued structures and constraint satisfaction problems.

\subsection{First-Order Logic}

A \emph{vocabulary} $\voc$ is a finite set of predicate and function symbols, each with an associated arity. We often denote a symbol $S$ with arity $n$ by $S/n$. A \emph{$\voc$-structure} $I$ consists of a domain $D$, an assignment of a relation $P^I \subseteq D^n$ to each predicate symbol $P/n \in \voc$, and an assignment of a function $F^I : D^n \to D$ to each function symbol $F/n \in \voc$. If $I$ is a $\voc$-structure and $\voc' \subseteq \voc$, we denote by $\res{I}{\voc'}$ the restriction of $I$ to the symbols of $\voc'$. If $\voc_1$ and $\voc_2$ are two disjoint vocabularies, $I$ a $\voc_1$-structure with domain $D$, and $J$ a $\voc_2$-structure with the same domain, then $I + J$ denotes the unique $(\voc_1 \cup \voc_2)$-structure with domain $D$ such that $\res{(I+J)}{\voc_1} = I$ and $\res{(I+J)}{\voc_2} = J$.

Variables are denoted by lowercase letters. We use $\xxx$, $\yyy$, \ldots, to denote both sets and tuples of variables. A \emph{variable assignment} with domain $D$ is a function mapping variables to domain elements in $D$. If $\theta$ is a variable assignment, $x$ a variable and $d$ a domain element, $\theta[x/d]$ denotes the variable assignment that maps $x$ to $d$ and corresponds to $\theta$ on all other variables. This notation is extended to tuples of variables and domain elements of the same length.

Terms and formulas over a vocabulary $\voc$ are defined as usual. We use $(\varphi \lir \psi)$ and $(\varphi \ler \psi)$ as shorthands for, respectively, the formulas $(\neg \varphi \lor \psi)$ and $((\varphi \lir \psi) \land (\psi \lir \varphi))$. If $\xxx$ and $\yyy$ are, respectively, the tuples of variables $(x_1,\ldots,x_n)$ and $(y_1,\ldots,y_n)$, then $\xxx \neq \yyy$ is a shorthand for the formula $(x_1 \neq y_1) \lor \ldots \lor (x_n \neq y_n)$. Often, we denote a formula $\varphi$ by $\varphi[\xxx]$ to indicate that $\xxx$ is precisely the set of free variables of $\varphi$. That is, if $y \in \xxx$, then $y$ has at least one occurrence in $\varphi$ outside the scope of quantifiers $\forall y$ and $\exists y$. A formula without free variables is called a \emph{sentence}. If $\varphi$ is a formula, $x$ a variable and $t$ a term, then $\varphi[x/t]$ denotes the result of replacing all free occurrences of $x$ in $\varphi$ by $t$. This notation is extended to tuples of variables and terms of the same length. We write $I\theta \models \varphi$ to say that a formula $\varphi$ evaluates to true in the structure $I$ under the variable assignment $\theta$. If all free variables of a formula $\varphi$ are among the set of variables $\xxx$, variable assignment $\theta$ is irrelevant in an expression of the form $I\theta[\xxx/\ddd] \models \varphi$, and therefore omitted. 

A \emph{query} is an expression of the form $\{ \xxx \mid \varphi[\yyy] \}$, where $\varphi$ is a formula and $\yyy \subseteq \xxx$. Such a query corresponds to the Boolean lambda expression $\lambda \xxx.\varphi[\yyy]$. The interpretation $\{ \xxx \mid \varphi[\yyy] \}^I$ of query $\{ \xxx \mid \varphi[\yyy] \}$ in structure $I$ is the set $\{ \ddd \mid I[\xxx/\ddd] \models \varphi \}$.

Two formulas $\varphi_1$ and $\varphi_2$ are \emph{equisatisfiable} if $\varphi_1$ is satisfiable iff $\varphi_2$ is satisfiable. Clearly, if $\varphi_1$ and $\varphi_2$ are logically equivalent, then they are also equisatisfiable. The following form of equivalence lies in between logical equivalence and equisatisfiability. 
\begin{definition}
Let $\voc_1$ and $\voc_2$ be two vocabularies that share a common subvocabulary $\voc$ and let $\varphi_1$ and $\varphi_2$ be sentences over, respectively, $\voc_1$ and $\voc_2$. Then $\varphi_1$ and $\varphi_2$ are \emph{$\voc$-equivalent} if for any $\voc$-structure $I$, there exists an expansion $M_1$ of $I$ to $\voc_1$ such that $M_1 \models \varphi_1$ iff there exists an expansion $M_2$ of $I$ to $\voc_2$ such that $M_2 \models \varphi_2$.
\end{definition}
The following proposition presents a method to rewrite sentences to $\voc$-equivalent sentences. This rewriting method is called \emph{predicate introduction}, and is applied in, e.g., the well-known \citeN{Tseitin68eng} transformation.
\begin{proposition}\label{prop:fopredintro}
Let $\varphi$ be a sentence over a vocabulary $\voc$ and let $\psi[\xxx]$ be a subformula of $\varphi$ with $n$ free variables. Let $P/n$ be a new predicate symbol and denote by $\varphi'$ the result of replacing $\psi[\xxx]$ by $P(\xxx)$ in $\varphi$. Then $\varphi' \land \forall \xxx (P(\xxx) \ler \psi[\xxx])$ is $\voc$-equivalent to $\varphi$.
\end{proposition}

In the rest of this paper, we facilitate the presentation by assuming that vocabularies do not contain function symbols. The following proposition sketches a method to remove function symbols from a theory.
\begin{proposition}\label{prop:fftheory}
Let $T$ be a theory over a vocabulary $\voc$. Then there exists a theory $T'$ over a function-free vocabulary $\voc'$ such that there is a one-to-one correspondence between the models of $T$ and the models of $T'$. Moreover, $T'$ can be constructed in linear time in the size of $T$.
\end{proposition}
The vocabulary $\voc'$ mentioned in the proposition can be obtained from $\voc$ by removing all function symbols and adding a new $(n+1)$-ary predicate symbol $P_F$ for each $n$-ary function symbol $F \in \voc$. Theory $T'$ is obtained from $T$ by adding the sentences 
\begin{align*}
	& \forall \xxx \exists y\ P_F(\xxx,y), \\
	& \forall \xxx \forall y_1 \forall y_2\ (P_F(\xxx,y_1) \land P_F(\xxx,y_2) \lir y_1 = y_2),
\end{align*}
for each of the introduced predicate symbols $P_F$, by moving all function symbols outside predicates using the standard equivalence-preserving rewrite rules, and finally replacing all atoms of the form $F(\xxx) = y$ by $P_F(\xxx,y)$. Let $I$ be a $\voc$-structure and $I'$ be the $\voc'$-structure defined by $P^{I'} = P^I$ for each $P \in \voc \cap \voc'$ and $P_F^{I'} = \{ \ddd,d' \mid F^I(\ddd) = d' \}$ for each function symbol $F \in \voc$. Then $I$ is a model of $T$ iff $I'$ is a model of $T'$. Moreover, each model of $T'$ can be obtained from a model of $T$ in this manner.
\begin{example}
Applying the sketched transformation on a theory containing the sentence $\id{Selected}(\id{C})$, produces a theory containing the sentences
\begin{align*}
	& \exists y\ P_{\id{C}}(y); \\
	& \forall y_1 \forall y_2\ (P_{\id{C}}(y_1) \land P_{\id{C}}(y_2) \lir y_1 = y_2); \\
	& \forall x\ (P_{\id{C}}(x) \lir \id{Selected}(x)).
\end{align*}
\end{example}

\subsection{Three- and Four-Valued Structures}\label{sec:struct}

In this section we present three- and four-valued structures. In these structures it is possible to express partial and inconsistent information.
%

\subsubsection{Four-Valued Structures}

\citeN{Belnap77} introduced a four-valued logic with truth values \emph{true}, \emph{false}, \emph{unknown}, and
\emph{inconsistent} which we denote by, respectively, $\true$, $\false$, $\unkn$ and $\inco$.  For a truth value $\tval$, the \emph{inverse value} $\tval^{-1}$ is defined by $\true^{-1} = \false$, $\false^{-1} = \true$, $\unkn^{-1} =
\unkn$ and $\inco^{-1} = \inco$. Belnap distinguished two orders, the \emph{truth order} $<_t$ and the \emph{precision order}
$<_p$, also called \emph{knowledge order}. They are defined in Figure~\ref{fig:order}. The reflexive closure of these orders is denoted by $\leqt$, respectively $\leqp$.
\begin{figure}
\begin{center}
\resizebox{0.25\textwidth}{!}{ 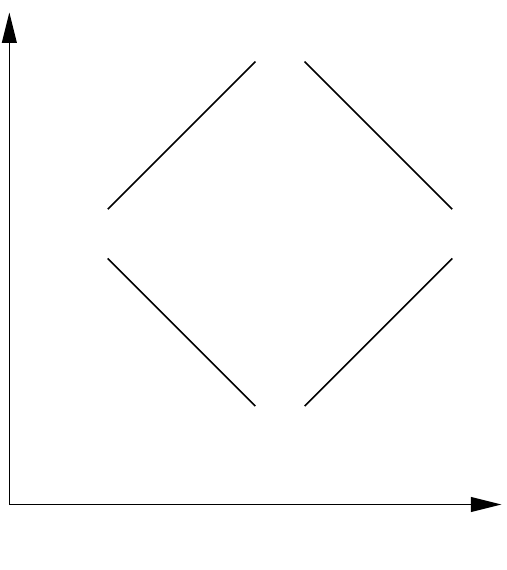 }
\end{center}
\caption{The truth and precision order. According to the truth axis, we have, e.g., $\false  <_t  \unkn$; according to the precision axis, we have, e.g., $\unkn <_p \false$. }
\label{fig:order}
\end{figure}

Let $\voc$ be a (function-free) vocabulary. A four-valued $\voc$-structure $\app{I}$ consists of a domain $D$ and a function $P^{\app{I}}: D^n \to \{ \true,\false,\unkn,\inco \}$ for every $P/n \in \voc$. We say that a four-valued structure $\app{I}$ is \emph{three-valued} when $P^{\app{I}}(\ddd) \neq \inco$ for any $P \in \voc$ and tuple of domain elements $\ddd$. A structure $\app{I}$ is \emph{two-valued} when it is three-valued and $P^{\app{I}}(\ddd) \neq \unkn$ for every $P$ and $\ddd$. We call a four-valued structure $\app{I}$ \emph{strictly three-valued} if it is three-valued but not two-valued. Likewise, a structure is \emph{strictly four-valued} if it is four-valued but not three-valued. A four-valued structure $\app{I}$ that is two-valued can be identified with the standard FO structure $I$ for which for every predicate symbol $P$ and tuple of domain elements $\ddd$, $\ddd \in P^I$ iff $P^{\app{I}}(\ddd) = \true$. In the rest of the paper, when we refer to a structure $I$ (without tilde) we mean a two-valued structure, while $\app{I}$ means a four-valued structure (which possibly is three- or two-valued).

The precision order extends to structures: if $\app{I}$ and $\app{J}$ are two $\voc$-structures, $\app{I} \leqp \app{J}$ if for every predicate symbol $P$ and tuple of domain elements $\ddd$, $P^{\app{I}}(\ddd) \leqp P^{\app{J}}(\ddd)$. Similarly, the truth order is extended to structures. 

The most precise $\voc$-structure with domain $D$ is denoted by
$\mapp_{\voc,D}$ and assigns $P^{\mapp_{\voc,D}}(\ddd) = \inco$ to
every predicate symbol $P/n \in \voc$ and $\ddd \in D^n$. Vice versa,
the least precise structure $\lapp_{\voc,D}$ assigns
$P^{\lapp_{\voc,D}}(\ddd) = \unkn$. We omit $D$ and/or $\voc$ from $\lapp_{\voc,D}$ and $\mapp_{\voc,D}$ if they are clear from the context. If a two-valued structure $I$ is more precise than a three-valued structure $\app{I}$, we say that $\app{I}$ \emph{approximates} $I$. 

The \emph{size} $\size{\app{I}}$ of a structure $\app{I}$ is defined as the cardinality of the domain of $\app{I}$. This definition is precise enough for the complexity results in this paper.

A \emph{domain atom} over a structure $\app{I}$ with domain $D$ is a pair of an $n$-tuple $\ddd$ of domain elements and an $n$-ary predicate symbol. We denote such a domain atom by $P(\ddd)$. For a truth value $\tval$ and domain atom $P(\ddd)$, we denote by $\app{I}[P(\ddd)/\tval]$ the structure that assigns $P^{\app{I}}(\ddd) = \tval$ and corresponds to $\app{I}$ on the rest of the vocabulary. 

A \emph{domain literal} is a domain atom $P(\ddd)$ or the negation $\neg P(\ddd)$ of a domain atom. By $\app{I}[\neg P(\ddd)/\tval]$ we denote the structure $\app{I}[P(\ddd)/\tval^{-1}]$. This notation is extended to sets of domain literals: if $U$ is the set $\{ L_1,\ldots,L_n \}$ of domain literals, $\app{I}[U/\tval]$ denotes the structure $\app{I}[L_1/\tval]\cdots[L_n/\tval]$.

The value of a formula $\varphi$ in a four-valued structure $\app{I}$ with domain $D$ under variable assignment $\theta$ is defined by structural induction: 
\begin{longitem}
	\item $\app{I}\theta(P(\xxx)) = P^{\app{I}}(\theta(\xxx))$;
	\item $\app{I}\theta(\neg \varphi) = (\app{I}\theta(\varphi))^{-1}$;
	\item $\app{I}\theta(\varphi \land \psi) = \glb_{\leqt}\{ \app{I}\theta(\varphi), \app{I}\theta(\psi) \}$;
	\item $\app{I}\theta(\varphi \lor \psi) = \lub_{\leqt}\{ \app{I}\theta(\varphi), \app{I}\theta(\psi) \}$;
	\item $\app{I}\theta(\forall x\ \varphi) = \glb_{\leqt}\{ \app{I}\theta[x/d](\varphi) \mid d \in D \}$;
	\item $\app{I}\theta(\exists x\ \varphi) = \lub_{\leqt}\{ \app{I}\theta[x/d](\varphi) \mid d \in D \}$.
\end{longitem}
When $\app{I}$ is a three-valued structure, this corresponds to the standard \citeN{Kleene52} semantics.

If $\app{I}$ is three-valued, then $\app{I}\theta(\varphi) \neq \inco$ for every formula $\varphi$ and variable assignment $\theta$.  If $\app{I}$ is two-valued, then $\app{I}\theta(\varphi) \in \{ \true, \false \}$. Also, if $\app{I}$ is two-valued, then $\app{I}\theta(\varphi) = \true$ iff $\app{I}\theta \models \varphi$. We omit $\theta$ and/or $[\xxx/\ddd]$ from an expression of the form $\app{I}\theta[\xxx/\ddd](\varphi)$ when they are irrelevant.
 
If $\varphi$ is a formula and $\app{I}$ and $\app{J}$ are two structures such that $\app{I} \leqp \app{J}$, then also $\app{I}\theta(\varphi) \leqp \app{J}\theta(\varphi)$ for every $\theta$. If $\varphi$ is a formula that does not contain negations and $\app{I} \leqt \app{J}$, then also $\app{I}\theta(\varphi) \leqt \app{J}\theta(\varphi)$ for every $\theta$. 

Four-valued structures can be defined over vocabularies containing function symbols. For each such a structure $\app{I}$, there exists a structure $\app{I}'$ over a function-free vocabulary such that there is a one-to-one correspondence between the two-valued structures approximated by $\app{I}$ and the two-valued structures approximated by $\app{I}'$. In combination with Proposition~\ref{prop:fftheory}, this allows to apply all results in the rest of this paper in a context where function symbols are present. We refer the reader to~\cite{phd/Wittocx10} for details.

\subsubsection{Encoding Four-Valued Structures by Two-Valued Structures}\label{ssec:pairs}

A standard way to encode a four-valued structure $\app{I}$ over a vocabulary $\voc$ is by a two-valued structure $\tf{\app{I}}$ over a vocabulary $\tf{\voc}$ containing two symbols $\ct{P}$ and $\cf{P}$ for each symbol $P \in \voc$. The interpretation of $\ct{P}$, respectively $\cf{P}$, in $\tf{\app{I}}$ represents what is certainly true, respectively certainly false, in $P^{\app{I}}$. Formally, for a vocabulary $\voc$ and $\voc$-structure $\app{I}$, $\tf{\voc}$ denotes the vocabulary $\{ \ct{P}/n \mid P/n \in \voc \} \cup \{ \cf{P}/n \mid P/n \in \voc \}$ and the $\tf{\voc}$-structure $\tf{\app{I}}$ is defined by $\ct{P}^{\tf{\app{I}}} = \{ \ddd \mid P^{\app{I}}(\ddd) \geq_p \true \}$ and $\cf{P}^{\tf{\app{I}}} = \{ \ddd \mid P^{\app{I}}(\ddd) \geq_p \false \}$ for every $P \in \voc$. 

Observe that $\app{I}$ is three-valued iff $\ct{P}^{\tf{\app{I}}}$ and $\cf{P}^{\tf{\app{I}}}$ are disjoint for any $P \in \voc$; $\app{I}$ is two-valued iff for every $P/n \in \voc$, $\ct{P}^{\tf{\app{I}}}$ and $\cf{P}^{\tf{\app{I}}}$ are each others complement in $D^n$. Also, if $\app{I} \leqp \app{J}$, then $\ct{P}^{\tf{\app{I}}} \subseteq \ct{P}^{\tf{\app{J}}}$ and $\cf{P}^{\tf{\app{I}}} \subseteq \cf{P}^{\tf{\app{J}}}$. Therefore $\app{I} \leqp \app{J}$ iff $\tf{\app{I}} \leqt \tf{\app{J}}$.

The value of a formula $\varphi$ in a structure $\app{I}$ can be obtained by computing the value of two formulas over $\tf{\voc}$ in $\tf{\app{I}}$. Define for a formula $\varphi$ over $\voc$ the formulas $\ct{\varphi}$ and $\cf{\varphi}$ over $\tf{\voc}$ by simultaneous induction:
\begin{longitem}
	\item $\ct{(P(\xxx))} = \ct{P}(\xxx)$ and $\cf{(P(\xxx))} = \cf{P}(\xxx)$;
	\item $\ct{(\neg \varphi)} = \cf{\varphi}$ and $\cf{(\neg \varphi)} = \ct{\varphi}$;
	\item $\ct{(\varphi \land \psi)} = \ct{\varphi} \land \ct{\psi}$ and $\cf{(\varphi \land \psi)} = \cf{\varphi} \lor \cf{\psi}$;
	\item $\ct{(\varphi \lor \psi)} = \ct{\varphi} \lor \ct{\psi}$ and $\cf{(\varphi \lor \psi)} = \cf{\varphi} \land \cf{\psi}$;
	\item $\ct{(\forall x\ \varphi)} = \forall x\ \ct{\varphi}$ and $\cf{(\forall x\ \varphi)} = \exists x\ \cf{\varphi}$;
	\item $\ct{(\exists x\ \varphi)} = \exists x\ \ct{\varphi}$ and $\cf{(\exists x\ \varphi)} = \forall x\ \cf{\varphi}$.
\end{longitem}
The intuition is that $\ct{\varphi}$ denotes a formula that is true iff $\varphi$ is certainly true while $\cf{\varphi}$ is a formula that is true iff $\varphi$ is certainly false. This explains, e.g., the definition $\ct{(\neg\varphi)} = \cf{\varphi}$: $\neg\varphi$ is certainly true iff $\varphi$ is certainly false. As another example, $\cf{(\varphi \land \psi)} = \cf{\varphi} \lor \cf{\psi}$ states that $(\varphi \land \psi)$ is certainly false if $\varphi$ or $\psi$ is certainly false. 

For a pair of formulas $(\varphi_1,\varphi_2)$, a structure $I$ and variable assignment $\theta$, we denote the pair of truth values $(I\theta(\varphi_1),I\theta(\varphi_2))$ by $I\theta(\varphi_1,\varphi_2)$. We identify the pairs $(\true,\false)$, $(\false,\true)$, $(\false,\false)$ and $(\true,\true)$ with, respectively, the truth values $\true$, $\false$, $\unkn$ and $\inco$. Intuitively, the first value in the pairs states whether something is certainly true, the second value whether it is certainly false. It follows that, e.g., $(\true,\false)$ corresponds to saying that something is certainly true and not certainly false and therefore identifies with $\true$. Using these equalities, the next proposition expresses that the value of a formula in a four-valued structure $\app{I}$ can be computed by evaluating $\ct{\varphi}$ and $\cf{\varphi}$ in the two-valued structure $\tf{\app{I}}$.

\begin{proposition}[\cite{jsyml/Feferman84}]\label{prop:ftot}
For every formula $\varphi$, structure $\app{I}$, and variable assignment $\theta$, $\app{I}\theta(\varphi) = \tf{\app{I}}\theta(\ct{\varphi}, \cf{\varphi})$.
\end{proposition}

It follows from Proposition~\ref{prop:ftot} and from the fact that it can be decided in polynomial time in $\size{I}$ whether a finite two-valued structure $I$ satisfies a formula, that $\app{I}\theta(\varphi)$ can be computed in polynomial time in $\size{\app{I}}$ for any finite four-valued structure $\app{I}$. Another interesting property of the formulas $\ct{\varphi}$ and $\cf{\varphi}$ is stated in the following proposition.
\begin{proposition}\label{prop:tfpos}
For every formula $\varphi$,  neither $\ct{\varphi}$ nor $\cf{\varphi}$ contain a `$\neg$'.
\end{proposition}

\subsection{Constraint Programming}

We now recall some definitions from Constraint Programming (CP). Let $S$ be a sequence $(v_1,\ldots,v_n)$ of variables. A \emph{constraint} on $S$ is a set of $n$-tuples. A \emph{constraint satisfaction problem} (CSP) is a tuple $\langle \mathcal{C}, V, \domf \rangle$ of a set $V$ of variables, a mapping $\domf$ of variables in $V$ to domains, and a set $\mathcal{C}$ of constraints on finite sequences of variables from $V$. A \emph{solution} to $\langle \mathcal{C}, V, \domf \rangle$ is a function $d$, mapping each variable $v$ of $V$ to a value $d(v) \in \dom{v}$ such that $(d(v_1),\ldots,d(v_n)) \in C$ for each constraint $C \in \mathcal{C}$ on sequence $(v_1,\ldots,v_n)$. Two CSPs sharing the same variables are called \emph{equivalent} if they have the same solutions.

A \emph{propagator} (also called a \emph{constraint solver}) is a function mapping CSPs to equivalent CSPs. A propagator is called \emph{domain reducing} if it retains the constraints of a CSP and does not increase its domains. That is, if propagator $\prp$ is domain reducing and $\prp(\langle \mathcal{C}_1,V,\domf_1\rangle) = \langle \mathcal{C}_2,V,\domf_2\rangle$, then $\mathcal{C}_2 = \mathcal{C}_1$ and for every $v \in V$, $\domf_2(v) \subseteq \domf_1(v)$. In this paper, we only consider domain reducing propagators.

We refer to the book of \citeN{Apt03} for a comprehensive introduction to CP. To avoid confusion between variables in the context of FO and variables of a CSP, we call the latter \emph{constraint variables} in the rest of this paper.\footnote{In fact, constraint variables correspond to $0$-ary function symbols, i.e., constant symbols, in the context of FO.} 

\section{Constraint Propagation for Logic Theories}\label{sec:cpforlogic}

In this section, we transfer some terminology and well-known results from CP to (first-order) logic theories. We rely on the property that for every pair of a finite structure $\app{I}$ and a theory $T$, there exists a CSP $\mathcal{P}$ such that there is a one-to-one correspondence between the models of $T$ approximated by $\app{I}$ and the solutions of $\mathcal{P}$.\footnote{The inverse property also holds: for each finite CSP $\langle \mathcal{C}, V, \domf \rangle$, there exists a first-order logic theory $T$ and finite structure $\app{I}$ such that there is a one-to-one correspondence between the solutions of the CSP and the models of $T$ approximated by $\app{I}$. The theory $T$ and structure $\app{I}$ can be constructed by introducing a constant $c_v$ for every $v \in V$ and a predicate $P_C$ for every $C \in \mathcal{C}$. Theory $T$ is then defined by $\{ P_C(c_{v_1},\ldots,c_{v_n}) \mid \text{$C \in \mathcal{C}$ is a constraint on $(v_1,\ldots,v_n)$ }\}$, while $\app{I}$ assigns $C$ to $P_C$ and allows each $c_v$ to take a value in $\domf(v)$~\cite[page 97]{phd/Wittocx10}.}

\subsection{From a Model Generation Problem to a CSP}\label{ssec:mgtocsp}

For the rest of this section, let $T$ be a logic theory over vocabulary $\voc$ and $\app{I}$ a four-valued $\voc$-structure with domain $D$. If $\app{I}$ is finite, then the pair $\langle T,\app{I} \rangle$ has a corresponding CSP which is denoted by $\langle \mathcal{C}_T,V_{\app{I}},\domf_{\app{I}} \rangle$ and defined as follows. The set of constraint variables $V_{\app{I}}$ is defined as the set of all domain atoms over $\voc$ and $D$. We assume a fixed total order on $V_{\app{I}}$ and call the $i$th element in that order the $i$th domain atom. The domain $\domf_{\app{I}}(P(\ddd))$ associated to domain atom $P(\ddd)$ is defined by
\[
	\domf_{\app{I}}(P(\ddd)) = 
		\begin{cases}
			\{ \true, \false \} & \text{if $P^{\app{I}}(\ddd) = \unkn$,} \\
			\{ \true \}			& \text{if $P^{\app{I}}(\ddd) = \true$,} \\
			\{ \false \}		& \text{if $P^{\app{I}}(\ddd) = \false$,} \\
			\emptyset			& \text{if $P^{\app{I}}(\ddd) = \inco$.} \\
		\end{cases}
\]
Given a tuple $\overline{\tval} \in \{ \true, \false \}^{\size{V_{\app{I}}}}$, $I_{\overline{\tval}}$ denotes the $\voc$-structure with domain $D$ such that for every $P$ and $\ddd$, $\ddd \in P^{I_{\overline{\tval}}}$ iff $P(\ddd)$ is the $i$th domain atom and the $i$th truth value in $\overline{\tval}$ is $\true$. Finally, $\mathcal{C}_T$ is the singleton set containing the constraint that consists of the set of tuples $\overline{\tval} \in \{ \true, \false \}^{\size{V_{\app{I}}}}$ such that $I_{\overline{\tval}} \models T$. It follows immediately that $\overline{\tval}$ is a solution to $\langle \mathcal{C}_T,V_{\app{I}},\domf_{\app{I}} \rangle$ iff $I_{\overline{\tval}} \models T$ and $\app{I} \leqp I_{\overline{\tval}}$.

\subsection{Propagators}

A structure $\app{I}$ can be seen as approximating some models of $T$, namely all two-valued structures $M$ such that $\app{I} \leqp M$ and $M \models T$. The goal of constraint propagation for $T$ is then to find a better approximation of these models, i.e., one that is more precise than $\app{I}$. We call an operator on the class of four-valued $\voc$-structures a \emph{propagator for $T$} if it performs constraint propagation for $T$. Formally, $\prp$ is a propagator for $T$ if the following two conditions are met: 
\begin{enumerate}
	\item $\prp$ is \emph{inflationary} with respect to $\leqp$. That is, $\app{I} \leqp \prp(\app{I})$ for every structure $\app{I}$. 	
	\item For every model $M$ of $T$ such that $\app{I} \leqp M$, also $\prp(\app{I}) \leqp M$. 
\end{enumerate}
The first condition states that by applying an operator no information is lost. The second condition states that no models of $T$ approximated by $\app{I}$ are lost. Note that for a propagator $\prp$ it follows from the definition above that $\app{I}$ and $\prp(\app{I})$ must have the same domain.

The following proposition relates the definition of a propagator for $T$ to the definition of propagator in the context of CP. The proof of the proposition is straightforward.
\begin{proposition}\label{prop:cpandfo}
Let $\prp$ is a propagator for $T$, $D$ a finite set, and $C$ be the class of CSPs of the form $\langle \mathcal{C}_T,V_{\app{I}},\domf_{\app{I}} \rangle$, where $\app{I}$ has domain $D$. Then the operator $f$ on $C$ defined by $f(\langle \mathcal{C}_T,V_{\app{I}},\domf_{\app{I}} \rangle) = \langle \mathcal{C}_T,V_{\app{I}},\domf_{\prp(\app{I})} \rangle$, is a domain reducing propagator.
\end{proposition}

A propagator $\prp$ is called \emph{monotone} if for every two structures $\app{I}$ and $\app{J}$ such that $\app{I} \leqp \app{J}$, also $\prp(\app{I}) \leqp \prp(\app{J})$ holds. An example of a monotone propagator is the \emph{inconsistency propagator} $\iprop$, defined by 
\[
	\iprop(\app{I}) = 
	\begin{cases}
		\app{I}	& \text{if $\app{I}$ is three-valued} \\
		\mapp		& \text{otherwise}
	\end{cases}
\]
A propagator $\prp$ for $T$ is \emph{inducing} for $T$ if for every two-valued structure $I$ such that $I \not\models T$, $\prp(I)$ is strictly four-valued, i.e., it recognizes that $I$ is not a model and assigns $\inco$ to at least one domain element.

Note that the composition of two propagators is a propagator itself.
\begin{lemma}\label{lem:propcompose}
If $T_1$ and $T_2$ are theories over the same vocabulary, $\prp_1$ is a propagator for $T_1$ and $\prp_2$ a propagator for $T_2$, then $\prp_1 \circ \prp_2$ is a propagator for $T_1 \cup T_2$.
\end{lemma}
It is easy to check that the composition of two monotone propagators is a monotone propagator. Also, if $\prp_1$ is inducing for $T_1$ and $\prp_2$ is inducing for $T_2$, then $\prp_1 \circ \prp_2$ is inducing for $T_1 \cup T_2$.

From the definition of propagator, it follows that two logically equivalent theories have the same propagators. For $\voc$-equivalent theories, we have the following property.
\begin{proposition}\label{prop:extendprop}
Let $\voc$ and $\voc'$ be two vocabularies such that $\voc \subseteq \voc'$ and let $T$ and $T'$ be $\voc$-equivalent theories over $\voc$, respectively $\voc'$. Let $O'$ be an operator on $\voc'$ structures and define the operator $O$ on $\voc$-structures by $O(\app{I}) = \res{O'(\app{I}+\lapp_{\voc'\setminus\voc})}{\voc}$ for any $\voc$-structure $\app{I}$. If $O'$ is a propagator for $T'$, then $O$ is a propagator for $T$. If $O'$ is monotone, then $O$ is monotone as well.
\end{proposition}

\subsection{Refinement Sequences}

If $V$ is a set of propagators for $T$, Lemma~\ref{lem:propcompose} ensures that constraint propagation for $T$ can be performed by starting from $\app{I}$ and successively applying propagators from $V$. We then get a sequence of increasingly precise four-valued structures. If such a sequence is strictly increasing in precision, we call it a \emph{$V$-refinement sequence from $\app{I}$}.

\begin{definition}
Let $V$ be a set of propagators for $T$. We call a (possibly transfinite) sequence $\langle \app{J}_{\xi} \rangle_{0 \leq \xi \leq \alpha}$ of four-valued structures a \emph{$V$-refinement sequence from $\app{I}$} if 
\begin{longitem}
	\item $\app{J}_0 = \app{I}$,
	\item $\app{J}_{\xi+1} = \prp(\app{J}_{\xi})$ for some $\prp \in V$,
	\item $\app{J}_{\xi} <_p \app{J}_{\xi+1}$ for every $0 \leq \xi < \alpha$, 
	\item and $\app{J}_{\lambda} = \lub_{\leqp}(\{ \app{J}_{\xi} \mid \xi < \lambda \})$ for every limit ordinals $\lambda \leq \alpha$.
\end{longitem}
\end{definition}

In the CP literature, refinement sequences are sometimes called \emph{derivations}, and constructing a derivation is called \emph{constraint propagation}. Since refinement sequences are strictly increasing in precision, it follows that every refinement sequence from a finite structure $\app{I}$ is finite. Moreover:

\begin{proposition}\label{prop:len}\label{PROP:LEN}
For any fixed set of propagators $V$, the length of a $V$-refinement sequence from a finite structure $\app{I}$ is polynomial in $\size{\app{I}}$.
\end{proposition}

A refinement sequence is \emph{stabilizing} if it cannot be extended anymore. The last structure in a stabilizing refinement sequence is called the \emph{limit} of the sequence. A well-known result (see, e.g., Lemma 7.8 in \cite{Apt03}) states: 
\begin{proposition}\label{prop:refs}
Let $V$ be a set of monotone propagators for $T$ and let $\app{I}$ be a structure. Then every stabilizing $V$-refinement sequence from $\app{I}$ has the same limit.
\end{proposition}

If $V$ only contains monotone propagators, we denote by $\lim_V$ the operator that maps every finite structure to the unique limit of any stabilizing $V$-refinement sequence from finite structure $\app{I}$. From Lemma~\ref{lem:propcompose} it follows that $\lim_V$ is a propagator.

Besides monotonicity, other properties of propagators, e.g., idempotence, may be taken into account by algorithms to efficiently construct refinement sequences. \citeN{tcs/Apt99} provides a general overview of such properties and algorithms.

\subsection{Complete Propagators}

The \emph{complete propagator} for a theory $T$ is the propagator that yields the most precise structures. This propagator is denoted by $\rfi{T}$ and defined by 
\[ \rfi{T}(\app{I}) = \glb_{\leqp}\left(\{ M \mid \text{$\app{I} \leqp M$ and $M \models T$} \}\right). \]
The following properties hold for $\rfi{T}$:
\begin{proposition}\label{lem:complete}
For every theory $T$, $\rfi{T}$ is a monotone propagator.
\end{proposition}
\begin{proposition}\label{prop:opt}
Let $\prp$ be a propagator for $T$ and $\app{I}$ a structure. Then
$\prp(\app{I}) \leqp \rfi{T}(\app{I})$. That is, $\rfi{T}$ is the most precise propagator.
\end{proposition}

\begin{example}\label{ex:stud}
Let $\voc = \{ \id{Module}/1, \id{Selected}/1, \id{In}/2, \id{MutExcl}/2 \}$ and let $\app{I}_0$ be the $\voc$-structure with domain $\{ m_1, m_2, c_1, c_2, c_3, c_4 \}$ that is two-valued on all symbols except $\id{Selected}$, and that is given by
\begin{align*}
	& \id{Module}^{\app{I}_0} = \{ m_1, m_2\}, & &  \id{MutExcl}^{\app{I}_0} = \{ (c_1,c_2) \}, \\
	& \ct{\id{Selected}}^{\tf{\app{I}_0}} = \{ c_1 \}, & & \cf{\id{Selected}}^{\tf{\app{I}_0}} = \emptyset, \\
	&   \id{In}^{\app{I}_0} = \{ (c_1,m_1), (c_3,m_1), (c_2,m_2) \}. & &  
\end{align*}
This structure expresses that course $c_1$ is certainly selected, while it is unknown whether other modules or courses are selected. Let $T_1$ be the theory that consists of the sentences \eqref{form:cc}--\eqref{form:ac} from the introduction.  Then structure $\rfi{T_1}(\app{I_0})$ assigns $\ct{Selected}^{\tf{\rfi{T_1}(\app{I_0})}} = \{ m_1, c_1, c_3 \}$ and $\cf{Selected}^{\tf{\rfi{T_1}(\app{I_0})}} = \{ m_2, c_2 \}$. Indeed, because $c_1$ is selected according to $\app{I}_0$, we can derive from \eqref{form:cc} that $c_2$ cannot be selected. Next,~\eqref{form:ac} implies that module $m_2$ cannot be selected. It then follows from~\eqref{form:om} that $m_1$ must be selected. This implies in turn that $c_3$ must be selected. No information about $c_4$ can be derived since both $\rfi{T_1}(\app{I_0})[\id{Selected}(c_4)/\true]$ and $\rfi{T_1}(\app{I_0})[\id{Selected}(c_4)/\false]$ are models of $T_1$.
\end{example}

Observe that if $T$ has no models approximated by $\app{I}$, then $\rfi{T}(\app{I}) = \mapp$. Note that this is the case if $\app{I}$ is strictly four-valued. Therefore, the problem of deciding whether a given domain atom is inconsistent in $\rfi{T}(\app{I})$ is at least as hard as deciding whether $T$ has a model approximated by $\app{I}$. If $T$ is an FO theory, the latter problem is intractable: for a fixed $T$ and varying finite structures $\app{I}$ it is \np-complete~\cite{fagin74}, for infinite structures $\app{I}$, it is undecidable. Consequently, computing $\rfi{T}(\app{I})$ for a fixed FO theory $T$ and varying finite structures $\app{I}$ is intractable.

Similarly as for theories, we associate to each sentence $\varphi$ the monotone propagator $\rfi{\varphi}$, which maps a structure $\app{I}$ to the most precise approximation of $\varphi$ from $\app{I}$. That is, 
\[ \rfi{\varphi}(\app{I}) = \glb_{\leqp}\left(\{ M \mid \text{$\app{I} \leqp M$ and $M \models \varphi$}\}\right). \]
Observe that for any sentence $\varphi$ implied by $T$, $\rfi{T}(\app{I})$ is more precise than $\rfi{\varphi}(\app{I})$, since 
\[ \{ J \mid \text{$\app{I} \leqp J$ and $J \models T$} \} \subseteq \{ J \mid \text{$\app{I} \leqp J$ and $J \models \varphi$}\} \]
As such, we obtain the following proposition.
\begin{proposition}\label{lem:fisp}
If $T \models \varphi$, then $\rfi{\varphi}$ is a monotone propagator for $T$.
\end{proposition}
In particular, if $\varphi \in T$, then $\rfi{\varphi}$ is a monotone propagator for $T$.

From Proposition~\ref{prop:refs} and Proposition~\ref{lem:fisp} it follows that every stabilizing $\{ \rfi{\varphi} \mid \varphi \in T \}$-refinement sequence from finite structure $\app{I}$ has the same limit. We denote the propagator $\lim_{\{\rfi{\varphi} \mid \varphi \in T\}}$ by $\trsl{T}$. We call a $\{ \rfi{\varphi} \mid \varphi \in T \}$-refinement sequence also a \emph{$T$-refinement sequence}. 

\begin{example}
Let $\app{I}_0$ and $T_1$ be as in Example~\ref{ex:stud}. Let $\langle \app{I}_i \rangle_{0 \leq i \leq 4}$ be the $T_1$-refinement sequence from $\app{I}_0$ obtained by applying (in this order) the propagators $\rfi{\eqref{form:cc}}$, $\rfi{\eqref{form:ac}}$, $\rfi{\eqref{form:om}}$ and $\rfi{\eqref{form:ac}}$. A reasoning similar to the one we made in Example~\ref{ex:stud} shows that $c_2 \in \cf{Selected}^{\tf{\app{I}_1}}$, $m_2 \in \cf{Selected}^{\tf{\app{I}_2}}$, $m_1 \in \ct{Selected}^{\tf{\app{I}_3}}$ and $c_3 \in \ct{Selected}^{\tf{\app{I}_4}}$. Hence, $\app{I}_4 = \rfi{T_1}(\app{I}_0)$, the refinement sequence is stabilizing and $\trsl{T_1}(\app{I}_0) = \rfi{T_1}(\app{I}_0)$.
\end{example}

\begin{example}\label{ex:rfitrsl}
Let $T$ be the theory $\{ P \ler Q, P \ler \neg Q \}$. Then $\rfi{T}(\lapp) = \mapp$ and $\trsl{T}(\lapp) = \lapp$.
\end{example}

As Example~\ref{ex:rfitrsl} shows, it is not necessarily the case that $\trsl{T}(\app{I}) = \rfi{T}(\app{I})$. In general, only $\trsl{T}(\app{I}) \leqp \rfi{T}(\app{I})$ holds. Note that  $\trsl{T}(\app{I}) = \rfi{T}(\app{I})$ holds if $T$ contains precisely one sentence. 

\section{Polynomial Propagation for First-Order Logic}\label{sec:foprop}

In the previous section, we introduced the idea of refining a four valued structure $\app{I}$ by propagation. In this section, we introduce a constraint propagation method for FO theories $T$ that is computationally less expensive than applying $\rfi{T}$ or computing a (stabilizing) $T$-refinement sequence. The method we propose is based on \emph{implicational normal form propagators} (INF propagators). These propagators have several interesting properties. First, they are monotone, ensuring that stabilizing refinement sequences constructed using only INF propagators have a unique limit. Secondly, INF propagators have polynomial-time data complexity and therefore stabilizing refinement sequences using only INF propagators can be computed in polynomial time. Thirdly, such a refinement sequence can be represented by a set of positive, i.e., negation-free, rules, which makes it possible to use, e.g., logic programming systems to compute the result of propagation. Finally, INF propagators can be applied in a symbolic way, i.e., independent of a four-valued input structure.

\subsection{Implicational Normal Form Propagators}

INF propagators are associated to FO sentences in implicational normal form.

\begin{definition}
An FO sentence is in \emph{implicational normal form (INF)} if it is of the form $\forall \xxx \ (\psi \lir L[\xxx])$, where $\psi$ is an arbitrary formula with free variables among $\xxx$ and $L$ a literal with free variables $\xxx$.
\end{definition}
For an INF sentence $\forall \xxx\ (\psi \lir L[\xxx])$, the associated INF propagator computes the value of $\psi$ in the given structure. If this value is $\true$ or $\inco$, the literal $L[\xxx]$ is made true (or inconsistent if it was false or inconsistent in the given structure). This is formalized in the following definition.
\begin{definition}\label{def:infp}
The operator $\infp{\varphi}$ associated to the sentence $\varphi := \forall \xxx\ (\psi \lir P(\xxx))$ is defined by 
\[
	Q^{\infp{\varphi}(\app{I})}(\ddd_x) = 
	\begin{cases}
		\lub_{\leqp}\{\true,P^{\app{I}}(\ddd_x)\} & \text{if $Q = P$ and $\app{I}[\xxx/\ddd_x](\psi) \geq_p \true$} \\
		Q^{\app{I}}(\ddd_x)								& \text{otherwise} \\
	\end{cases}
\]
The operator $\infp{\varphi}$ associated to the sentence $\varphi := \forall \xxx\ (\psi \lir \neg P(\xxx))$ is defined by 
\[
	Q^{\infp{\varphi}(\app{I})}(\ddd_x) = 
	\begin{cases}
		\lub_{\leqp}\{\false,P^{\app{I}}(\ddd_x)\} & \text{if $Q = P$ and $\app{I}[\xxx/\ddd_x](\psi) \geq_p \true$} \\
		Q^{\app{I}}(\ddd_x)								 & \text{otherwise} 
	\end{cases}
\]
\end{definition}

\begin{example}
Sentence~\eqref{form:ac} of the introduction is an INF sentence. Let $\app{I}$ be a structure such that $\id{Course}^{\app{I}}(c_1)$, $\id{Module}^{\app{I}}(m_1)$, $\id{Selected}^{\app{I}}(m_1)$, and $\id{In}^{\app{I}}(c_1,m_1)$ are true. Then according to the definition of $\infp{\eqref{form:ac}}$, $\id{Selected}^{(\infp{\eqref{form:ac}}(\app{I}))}(c_1) \geqp \true$. That is, if module $m_1$ is certainly selected and course $c_1$ certainly belongs to $m_1$, then the operator $\infp{\eqref{form:ac}}$ associated to sentence~\eqref{form:ac} derives that $c_1$ is certainly selected. Note that this operator does
not perform contrapositive propagation. For instance, if $m_2$ is a module, $c_2$ a course, $\id{In}^{\app{I}}(c_2,m_2) = \true$ and $\id{Selected}^{\app{I}}(c_2) = \false$, the operator does not derive that $m_2$ is certainly not selected.
\end{example}

\begin{proposition}\label{prop:infpisprop}
For every INF sentence $\varphi$, $\infp{\varphi}$ is a monotone propagator.
\end{proposition}

As mentioned in Section~\ref{ssec:pairs}, evaluating a formula in a finite four-valued structure $\app{I}$ takes polynomial time in $\size{\app{I}}$. It follows that for a fixed INF sentence $\varphi$ and finite structure $\app{I}$, computing $\infp{\varphi}(\app{I})$ takes polynomial time in $\size{\app{I}}$. If we combine this result with Proposition~\ref{prop:len}, we obtain the following theorem.
\begin{theorem}\label{th:inf}
Let $V$ be a fixed finite set of INF sentences. Then $\lim_{\{ \infp{\varphi} \mid \varphi \in V\}}(\app{I})$ is computable in polynomial time in $\size{\app{I}}$ for every finite structure $\app{I}$.
\end{theorem}

\begin{proof}
Let $\varphi_1,\ldots,\varphi_n$ be all sentences in $V$. Let $\langle
\app{J}_i \rangle_{0 \leq i \leq m}$ be the longest sequence of structures such that $\app{J}_0 = \app{I}$ and $\app{J}_{i+1} = \infp{\varphi_k}(\app{J}_i)$, where $k$ is the lowest number between $1$ and $n$ such that $\app{J}_i \neq \infp{\varphi_k}(\app{J}_i)$. Clearly, $\langle \app{J}_i \rangle_{0 \leq i \leq m}$ is a stabilizing $\{\infp{\varphi} \mid \varphi \in V\}$-refinement sequence from $\app{I}$. Proposition~\ref{prop:len} implies that the length of this sequence is polynomial in $\size{\app{I}}$. Also, for each $i \geq 0$, $\app{J}_{i+1}$ can be computed in polynomial time in $\size{\app{I}}$: it suffices to compute $\infp{\varphi_1}(\app{J}_i)$, \ldots, $\infp{\varphi_n}(\app{J}_i)$, and each $\infp{\varphi_k}(\app{J}_i)$, $1 \leq k \leq n$, can be computed in polynomial time in $\size{\app{J}_i} = \size{\app{I}}$. Hence $\langle \app{J}_i \rangle_{0 \leq i \leq m}$ can be computed in polynomial time in $\size{\app{I}}$.
\end{proof}

\subsection{Representing INF Refinement Sequences by a Positive Rule Set}\label{ssec:itod}

For the rest of this section, let $V$ be a finite set of INF sentences and denote by $\infps{V}$ the set $\{ \infp{\varphi} \mid \varphi \in V \}$. We now show how to represent the propagator $\lim_{\infps{V}}$ by a set $\Delta$ of \emph{rules}, in the sense that for every structure $\lim_{\infps{V}}(\app{I})$ corresponds to a least model of $\Delta$. As mentioned in the introduction, the benefit of this representation is that sets of rules with a least model semantics are a basic component in many logic-based reasoning formalisms such as Prolog, Datalog, Stable Logic Programming and \foid. Hence, many of the theoretical and practical research results in these areas can be applied to study the properties of $\lim_{\infps{V}}$ and to easily obtain efficient implementations.

A \emph{rule set} over vocabulary $\voc$ is a finite set of rules of the form \[ \forall \xxx \ (P(\xxx) \rul \varphi[\yyy]), \] where $P \in \voc$, $\varphi$ is a formula over $\voc$, and $\yyy \subseteq \xxx$. The atom $P(\xxx)$ is called the \emph{head} of the rule, $\varphi$ the \emph{body}. The connective `$\rul$' is called \emph{definitional implication} and is to be distinguished from the connective `$\lir$'. A rule set $\Delta$ is \emph{positive} if none of the bodies in $\Delta$ contains a negation. $\Delta$ is \emph{monotone} if for every variable assignment $\theta$, every pair of structures $I$ and $J$ such that $I \leqt J$, and every rule body $\varphi$ of $\Delta$, $I\theta(\varphi) \leqt J\theta(\varphi)$.

The \emph{inflationary consequence operator} $\ico{\Delta}$ for positive definition $\Delta$ over $\voc$ is the operator on two-valued $\voc$-structures defined by $\ddd \in P^{\ico{\Delta}(I)}$ iff  $\ddd \in P^I$ or there exists a rule $\forall \xxx\ (P(\xxx) \rul \varphi)$ in $\Delta$ such that $I[\xxx/\ddd] \models \varphi$.

Note that the operator $\ico{\Delta}$ is $\leqt$-monotone. A structure $I$ satisfies $\Delta$, denoted $I \models \Delta$, if $I$ is a fixpoint of $\ico{\Delta}$. 

To each finite set of INF sentences over $\voc$, we associate the following positive rule set over $\tf{\voc}$. 
\begin{definition}\label{def:dvdi}
Let $V$ be a set of INF sentences and $\app{I}$ a structure. The \emph{rule set associated to $V$} is denoted by $\Delta_{V}$ and defined by 
\[ \Delta_{V} = \left\{ \forall \xxx\ (\ct{(L[\xxx])} \rul \ct{\psi}) \mid \forall \xxx \ (\psi \lir L[\xxx]) \in V \right\}. \]
\end{definition}

Observe that because of Proposition~\ref{prop:tfpos}, $\Delta_V$ is a positive rule set. The following proposition explains that $\Delta_V$ can be seen as a description of $\lim_{\infps{V}}$.
\begin{proposition}\label{prop:inp}\label{PROP:INPP}
For every set $V$ of INF sentences over $\voc$ and $\voc$-structure $\app{I}$, 
$\tf{\lim_{\infps{V}}(\app{I})} = \glb_{\leqt}(\{ M \mid M \models \Delta_V \text{ and } M \geqt \tf{\app{I}} \})$. 
\end{proposition}

Phrased differently, $\tf{\lim_{\infps{V}}(\app{I})}$ is the least Herbrand model of the positive rule set $\Delta$ obtained by introducing a fresh constant symbol $C_d$ for every domain element $d$ in $\app{I}$ and adding to $\Delta_V$ the rules $\ct{P}(C_{d_1},\ldots,C_{d_n}) \rul \top$, respectively $\cf{P}(C_{d_1},\ldots,C_{d_n}) \rul \top$, for every domain atom $P(d_1,\ldots,d_n)$ that is true, respectively false, in $\app{I}$. 

There are several benefits of using $\Delta_V$ as a description of $\lim_{\infps{V}}$. From a practical point of view, Proposition~\ref{prop:inp} states that we can use any existing algorithm that computes the least Herbrand model of positive rule sets to implement $\lim_{\infps{V}}$. Several such algorithms have been developed. For example, in the area of production rule systems, {\sc rete} \cite{ai/Forgy82} and {\sc leaps} \cite{aaai/MirankerBLG90} are two well-known algorithms; \citeN{IEEETKDE/VanWeert2010} describes improvements of these algorithms, used in implementations of Constraint Handling Rules. Other examples are the algorithms implemented in Prolog systems with tabling such as {\sc xsb} \cite{iclp/Swift09} and {\sc yap} \cite{jucs/SilvaC06}. In the context of databases, a frequently used algorithm is semi-naive evaluation \cite{Ullman88}. Adaptations of the semi-naive evaluation are implemented in the grounding component of \dlv \cite{amai/PerriSCL07} and in the grounder \gidl \cite{jair/WittocxMD10}. It follows that the large amount of research on optimization techniques and execution strategies for these algorithms can be used to obtain efficient implementations of $\lim_{\infps{V}}$ for a set $V$ of INF propagators. 

Most of the algorithms and systems mentioned above expect that all rules are of the form $\forall \xxx\ (P(\xxx) \rul \exists \yyy\ (Q_1(\zzz_1) \land \ldots \land Q_n(\zzz_n)))$, i.e., each body is the existential quantification of a conjunction of atoms. Some of the algorithms, e.g., semi-naive evaluation, can easily be extended to more general rule sets. Instead of extending the algorithms, one can as well rewrite rule sets into the desired format by applying predicate introduction \cite{VennekensMWD07a}, provided that only structures with finite domains are considered.

Other potential benefits of representing $\lim_{\infps{V}}$ by $\Delta_V$ stem from the area of logic program analysis. For instance, \emph{abstract interpretation} of logic programs \cite{jlp/Bruynooghe91} can be used to derive interesting properties of $\lim_{\infps{V}}$, \emph{program specialization} \cite{aicom/Leuschel97} to tailor $\Delta_V$ to a specific class of structures $\app{I}$, \emph{folding} \cite{PettorossiP98} to combine the application of several propagators, etc.

\subsection{From FO to INF}\label{ssec:fotoinf}

As mentioned above, computing $\rfi{T}(\app{I})$ can be computationally expensive. The same holds for $\trsl{T}(\app{I})$. For instance, if $T$ contains only one sentence, then $\rfi{T} = \trsl{T}$, and therefore the best known algorithms for applying $\trsl{T}$ can take exponential time in $\size{\app{I}}$ for finite structures $\app{I}$. In this section, we present a computationally cheaper method for constraint propagation on FO theories. The method consists of transforming, in linear time, a theory $T$ into an equivalent set of INF sentences. Then propagation on $T$ can be performed by applying the corresponding INF propagators. Theorem~\ref{th:inf} ensures that this propagation has polynomial-time data complexity. The price for this improved efficiency is of course a loss in precision. 

The next subsection describes the transformation. A diligent reader will note that the algorithm is non-deterministic and that a much more compact set of INF sentences can be generated (our implementation does). However, as we have no claims of optimality and our sole aim is to state polynomial-time data complexity, we present the most straightforward transformation.

\subsubsection{From FO to Equivalence Normal Form}

The transformation of theories to INF sentences works in two steps. First, a theory $T$ is transformed into a $\voc$-equivalent set of sentences in \emph{equivalence normal form} (ENF). Next, each ENF sentence is replaced by a set of INF sentences. We show that both steps can be executed in linear time. Also, we mention a theorem stating that under mild conditions, no precision is lost in the second step. That is, for each ENF sentence $\varphi$ that satisfies these conditions, there exists a set of INF sentences $V$ such that $\rfi{\varphi} = \lim_{\infps{V}}$. 

\begin{definition}\label{def:enfsentence}
An FO sentence $\varphi$ is in \emph{equivalence normal form (ENF)} if it is of the form $\forall \xxx\ (L[\xxx] \ler \psi[\xxx])$, where $\psi$ is of the form $(L_1 \land \ldots \land L_n)$, $(L_1 \lor \ldots \lor L_n)$, $(\forall \yyy\ L')$, or $(\exists \yyy\ L')$, and $L$, $L'$, $L_1$, \ldots, $L_n$ are literals.  An FO theory is in ENF if all its sentences are.
\end{definition}
Recall that we denote by $\psi[\xxx]$ that $\xxx$ are precisely the free variables of $\psi$. Thus, the definition of ENF implicitly states that in every ENF sentence $\forall \xxx\ (P(\xxx) \ler \psi)$, the free variables of $\psi$ are the free variables of $P(\xxx)$. Also, we allow that $n=1$ in the definition, i.e., $\forall \xxx\ (L[\xxx] \ler L_1[\xxx])$ is in ENF.

We now show that every FO theory $T$ over a vocabulary $\voc$ can be transformed into a $\voc$-equivalent ENF theory $T'$. The transformation is akin to the Tseitin transformation for propositional logic \cite{Tseitin68eng}. 
\begin{algotext}\label{alg:fotoenf}
Given an FO theory $T$: 
\begin{enumerate}
\item Push negation inside until they are directly in front
    of atoms (also eliminating implication) and ``flatten'' nested
    conjunctions and disjunctions, e.g., $((P \land Q) \land R)$ is
    replaced by $(P \land Q \land R)$.
	\item Replace every sentence $\varphi$ of $T$ that is not of the form $\forall \xxx\ (L[\xxx] \ler \psi[\xxx])$, where $L$ is a literal, by $\top \ler \varphi$.
	\item While $T$ is not in ENF:
	\begin{enumerate}
		\item Choose a sentence $\varphi$ of $T$ that is not in ENF. This sentence is of the form $\forall \xxx\ (L[\xxx] \ler \psi[\xxx])$.
		\item \label{step2b} Choose a direct subformula $\chi[\yyy]$ of $\psi$, replace $\chi[\yyy]$ by $\id{Aux}(\yyy)$ in $\psi$, where $\id{Aux}$ is a new predicate, and add the sentence $\forall \yyy\ (\id{Aux}(\yyy) \ler \chi[\yyy])$ to $T$.
	\end{enumerate}
	\item Return $T$.
\end{enumerate}
\end{algotext}
Clearly, the result of Algorithm~\ref{alg:fotoenf} is an ENF
theory.  Observe that the first step is linear in the size of $T$ and produces a theory $T'$ that is linear in the size of $T$. The auxiliary predicates introduced in step~\eqref{step2b}, replace subformulas of $T'$. Since the number of subformulas in $T'$ is linear in the size of $T$ and each subformula is replaced at most once by an auxiliary predicate, the algorithm runs in linear time. 
\begin{example}\label{ex:enf}
The result of applying Algorithm~\ref{alg:fotoenf} on the theory $T_1$ from Example~\ref{ex:stud} is the theory
\begin{align*}
	& \top \ler \forall x\forall y\ \id{Aux}_1(x,y), \\
	& \forall x\forall y\ (\id{Aux}_1(x,y) \ler \neg\id{MutExcl}(x,y) \lor \neg\id{Selected}(x) \lor \neg\id{Selected}(y)), \\
	& \top \ler \exists m\ \id{Aux}_2(m), \\
	& \forall m\ (\id{Aux}_2(m) \ler \id{Module}(m) \land \id{Selected}(m)), \\
	& \top \ler \forall c\ \id{Aux}_3(c), \\
	& \forall c\ (\id{Aux}_3(c) \ler \neg \id{Course}(c) \lor \id{Aux}_4(c) \lor \id{Selected}(c)), \\
	& \forall c\ (\id{Aux}_4(c) \ler \forall m\ \id{Aux}_5(m,c)), \\
	& \forall c \forall m\ (\id{Aux}_5(m,c) \ler \neg\id{Module}(m) \lor \neg\id{Selected}(m) \lor \neg\id{In}(c,m)).
\end{align*}
Here, the predicates $\id{Aux}_1$, \ldots, $\id{Aux}_5$ are introduced by the algorithm.
\end{example}

As steps (1) and (2) of Algorithm~\ref{alg:fotoenf} trivially preserve logical equivalence and step (3) preserves $\voc$-equivalence according to Proposition~\ref{prop:fopredintro}, the following proposition holds:
\begin{proposition}\label{prop:fotoinf}
Let $T'$ be the result of applying Algorithm~\ref{alg:fotoenf} on a theory $T$ over $\voc$. Then $T$ and $T'$ are $\voc$-equivalent.
\end{proposition}
The combination of Proposition~\ref{prop:extendprop} and Proposition~\ref{prop:fotoinf} ensures propagators for $T'$ can be used to implement propagators for $T$.

\subsubsection{From ENF to INF}

As shown in the previous section, every theory over $\voc$ can be transformed into a $\voc$-equivalent ENF theory. Now we show that any ENF theory $T$ can be transformed into a logically equivalent theory $\infs{T}$ containing only INF sentences. The transformation is inspired by standard rules for Boolean constraint propagation as studied, e.g., by \citeN{aaai/McAllester90} and \citeN{compulog/Apt99}. As a result, we obtain a propagator for $T$ with polynomial-time data complexity. Combined with the results of the previous section, this yields a propagation method for FO with polynomial-time data complexity. The relation of this propagation method to \emph{unit propagation} for propositional formulas in conjunctive normal form, is clarified by \citeN{compulog/Apt99} and \citeN{phd/Wittocx10}.

Each of the INF sentences we associate to an ENF sentence $\forall \xxx\ (L[\xxx] \ler \psi)$ is either logically equivalent to $\forall \xxx\ (L[\xxx] \lir \psi)$ or to $\forall \xxx\ (\psi \lir L[\xxx])$. The set of all INF sentences associated to an ENF sentence $\varphi$ contains for each predicate $P$ that occurs in $\varphi$ a sentence of the form $\forall \xxx\ (\psi \lir P(\xxx))$ and a sentence of the form $\forall \xxx\ (\psi \lir \neg P(\xxx))$. As such, the corresponding propagators are, in principle, able to derive that a domain atom $P(\ddd)$ is true, respectively false, if this is implied by $\varphi$.
\begin{definition}\label{def:etoi}
For an ENF sentence $\varphi$, the set of INF sentence $\infs{\varphi}$ is defined in Table~\ref{tab:etoi}. For an ENF theory $T$, $\infs{T}$ denotes the set of INF sentences $\bigcup_{\varphi \in T} \infs{\varphi}$.
\end{definition}
\begin{table}
\begin{center}
\begin{tabular}{l|l}
$\varphi$ & $\infs{\varphi}$ \\
\hline
\multirow{4}{*}{$\forall \xxx\ (L \ler L_1 \land \ldots \land L_n)$}
	& $\forall \xxx\ (L_1 \land \ldots \land L_n \lir L)$ \\
	& $\forall \xxx\ (\neg L_i \lir \neg L)$ \hfill $1 \leq i \leq n$\\                                                  
	& $\forall \xxx\ (L \lir L_i)$ \hfill $1 \leq i \leq n$\\                                                  
	& $\forall \xxx\ (\neg L \land L_1 \land \ldots \land L_{i-1} \land L_{i+1} \land \ldots \land L_n \lir \neg L_i)$ \ \hfill $1 \leq i \leq n$ \\
\hline
\multirow{4}{*}{$\forall \xxx\ (L \ler L_1 \lor \ldots \lor L_n)$} 
	& $\forall \xxx\ (\neg L_1 \land \ldots \land \neg L_n \lir \neg L)$ \\
	& $\forall \xxx\ (L_i \lir L)$ \hfill $1 \leq i \leq n$\\                                                  
	& $\forall \xxx\ (\neg L \lir \neg L_i)$ \hfill $1 \leq i \leq n$\\                                                  
	& $\forall \xxx\ (L \land \neg L_1 \land \ldots \land \neg L_{i-1} \land \neg L_{i+1} \land \ldots \land \neg L_n \lir  L_i)$ \ \hfill $1 \leq i \leq n$ \\
\hline
\multirow{4}{*}{$\forall \xxx\ (L[\xxx] \ler \forall \yyy\ L'[\xxx,\yyy])$} 
	& $\forall \xxx\ ((\forall \yyy\ L'[\xxx,\yyy]) \lir L[\xxx])$ \\
	& $\forall \xxx (\exists \yyy\ \neg L'[\xxx,\yyy]) \lir \neg L[\xxx])$ \\
	& $\forall \xxx \forall \yyy\ (L[\xxx] \lir L'[\xxx,\yyy])$ \\
	& $\forall \xxx \forall \yyy\ ((\neg L[\xxx] \land \forall \zzz\ (\yyy \neq \zzz \lir L'[\xxx,\yyy][\yyy/\zzz])) \lir \neg L'[\xxx,\yyy])$ \\
\hline
\multirow{4}{*}{$\forall \xxx\ (L[\xxx] \ler \exists \yyy\ L'[\xxx,\yyy])$} 
	& $\forall \xxx\ ((\forall \yyy\ \neg L'[\xxx,\yyy]) \lir \neg L[\xxx])$ \\
	& $\forall \xxx (\exists \yyy\ L'[\xxx,\yyy]) \lir L[\xxx])$ \\
	& $\forall \xxx \forall \yyy\ (\neg L[\xxx] \lir \neg L'[\xxx,\yyy])$ \\
	& $\forall \xxx \forall \yyy\ ((L[\xxx] \land \forall \zzz\ (\yyy \neq \zzz \lir \neg L'[\xxx,\yyy][\yyy/\zzz])) \lir L'[\xxx,\yyy])$
\end{tabular}
\caption{From ENF to INF.
}\label{tab:etoi}
\end{center}
\end{table}

It is straightforward to verify the following proposition.
\begin{proposition}
For every ENF sentence $\varphi$, $\infs{\varphi}$ is logically equivalent to $\varphi$. Similarly for ENF theories. 
\end{proposition}

It follows that if $T$ is an ENF theory, any propagator for $\infs{T}$ is a propagator for $T$. In particular, for every $\varphi \in \infs{T}$, $\infp{\varphi}$ is a polynomial-time propagator for $T$. As a corollary of Theorem~\ref{th:inf}, we have:
\begin{proposition}\label{prop:new}
If $T$ is an ENF theory, then the operator $\lim_{\infps{\infs{T}}}$ is a propagator for $T$. For a fixed ENF theory $T$ and varying finite structures $\app{I}$, $\lim_{\infps{\infs{T}}}(\app{I})$ can be computed in polynomial time.
\end{proposition}

\subsection{Summary}

Combining the results above yields the following propagation algorithm for FO theories.
\begin{algotext}\label{alg:prop}
For an input theory $T$ over $\voc$ and a $\voc$-structure $\app{I}$:
\begin{enumerate}
	\item Transform $T$ to an ENF theory $T'$ using Algorithm~\ref{alg:fotoenf}.
	\item Construct a (stabilizing) $\infps{\infs{T}}$-refinement sequence from $\app{I}$. Denote the last element by $\app{J}$. \label{it:propalgo}
	\item Return $\res{(\iprop(\app{J}))}{\voc}$.
\end{enumerate}
\end{algotext}
Note that this is an any-time algorithm: the refinement sequence constructed in the second step can be stabilizing, but this is not necessary. In either case, the algorithm implements a propagator for $T$. From Proposition~\ref{prop:new}, it follows that the algorithm has polynomial-time data complexity:
\begin{proposition}
For a fixed theory $T$ and varying finite structures $\app{I}$, Algorithm~\ref{alg:prop} can be implemented in polynomial time.
\end{proposition}

Since only INF propagators are used, the second step of Algorithm~\ref{alg:prop} can be executed by representing $\lim_{\infps{V}}$ as a positive rule set and computing the model of that set. In the following, we call Algorithm~\ref{alg:prop} the \emph{propagation algorithm}. For a general theory $T$, we denote by $\infs{T}$ the set of INF sentences $\infs{T'}$, where $T'$ is obtained from $T$ by applying Algorithm~\ref{alg:fotoenf}.

Algorithm~\ref{alg:prop} can be seen as an algorithm that propagates information up and down the parse tree of the input theory $T$. Indeed, let $\id{Aux}$ be a predicate and $\forall \xxx\ (\id{Aux}(\xxx) \ler \varphi)$ a sentence, introduced while transforming $T$ to ENF. As mentioned, $\id{Aux}$ represents the subformula $\varphi$ of $T$. Hence, INF sentences in $\infs{\forall \xxx\ (\id{Aux}(\xxx) \ler \varphi)}$ of the form $\forall \xxx\ (\psi \lir \id{Aux}(\xxx))$ or $\forall \xxx\ (\psi \lir \neg\id{Aux}(\xxx))$ propagate information derived about subformulas of $\varphi$ to $\varphi$ itself. That is, they propagate information upwards in the parse tree of $T$. The other INF sentences in $\infs{\forall \xxx\ (\id{Aux}(\xxx) \ler \varphi)}$ propagate information about $\varphi$ downwards. 

As an illustration, we apply the propagation algorithm on the theory and structure from Example~\ref{ex:stud}.
\begin{example}\label{ex:nonmin}
Let $T_1$ and $\app{I}_0$ be the theory and structure from Example~\ref{ex:stud}. Transforming $T_1$ to ENF produces the theory shown in Example~\ref{ex:enf}. According to Definition~\ref{def:etoi}, the set of INF sentences associated to this theory contains, amongst others, the sentences
\begin{align}
\label{eq:finp1}	& \forall x\forall y\ (\top \lir \id{Aux}_1(x,y)), \\
\label{eq:finp1-1}	& \forall x\forall y\ (\id{Aux}_1(x,y) \land \id{MutExcl}(x,y) \land \id{Selected}(x) \lir \neg\id{Selected}(y)), \\
\label{eq:finp2}	& \forall c\ (\top \lir \id{Aux}_3(c)), \\
\label{eq:finp3}	& \forall c\ (\id{Aux}_3(c) \land \id{Course}(c) \land \neg \id{Selected}(c) \lir \id{Aux}_4(c)), \\
\label{eq:finp4}	& \forall c\forall m\ (\id{Aux}_4(c) \lir \id{Aux}_5(m,c)), \\
\label{eq:finp5}	& \forall m\ (\id{Module}(m) \land \exists c\ (\id{Aux}_5(m,c) \land \id{In}(c,m)) \lir \neg \id{Selected}(m)), \\
\label{eq:finp6-1}	& \forall m\ (\neg \id{Selected}(m) \lir \neg \id{Aux}_2(m)), \\ 
\label{eq:finp6-2}	& \forall m\ (\neg \id{Module}(m) \lir \neg \id{Aux}_2(m)), \\ 
\label{eq:finp6}	& \forall m\ (\top \land \forall m'\ (m \neq m' \lir \neg \id{Aux}_2(m')) \lir \id{Aux}_2(m)), \\
\label{eq:finp6-3}	& \forall m\ (\id{Aux}_2(m) \lir \id{Selected}(m)), \\ 
\label{eq:finp7}	& \forall m \forall c\ (\id{Module}(m) \land \id{Selected}(m) \land \id{In}(c,m) \lir \neg \id{Aux}_5(m,c)), \\
\label{eq:finp8}	& \forall c \ (\exists m\ \neg \id{Aux}_5(m,c) \lir \neg \id{Aux}_4(c)), \\
\label{eq:finp9}	& \forall c \ (\id{Course}(c) \land \id{Aux}_3(c) \land \neg \id{Aux}_4(c) \lir \id{Selected}(c)). 
\end{align}
If one applies the associated INF propagators on $\app{I}_0$ in the order of the sentences above, the following information is derived. First, propagator $\infp{\eqref{eq:finp1-1}} \circ \infp{\eqref{eq:finp1}}$ derives that $\id{Aux}_1(c_1,c_2)$ is certainly true and that $c_2$ is certainly not selected. Next, $\infp{\eqref{eq:finp2}}$ derives that $\id{Aux}_3(c)$ is certainly true for all courses $c$. $\infp{\eqref{eq:finp3}}$ combines the derived information and concludes that $\id{Aux}_4(c_2)$ is certainly true. This in turn implies, by $\infp{\eqref{eq:finp4}}$, that $\id{Aux}_5(m,c)$ is certainly true for, a.o., $m=m_2$ and $c=c_2$. $\infp{\eqref{eq:finp5}}$ derives from the fact that $c_2$ belongs to $m_2$, that $m_2$ cannot be selected. Next, it is derived that $m_1$ is certainly selected by applying $\infp{\eqref{eq:finp6-3}} \circ \cdots \circ \infp{\eqref{eq:finp6-1}}$, and finally, applying $\infp{\eqref{eq:finp9}} \circ \infp{\eqref{eq:finp8}} \circ \infp{\eqref{eq:finp7}}$ yields that $c_3$ is certainly selected. As such, exactly the same information as in $\rfi{T_1}(\app{I}_0)$ is derived.
\end{example}

The following example gives another illustration of what the propagation algorithm can achieve.
\begin{example}\label{ex:plan}
Consider the theory $T_2$, taken from some planning domain, consisting of the sentence
\begin{align*}
 & \forall a \forall a_p \forall t\ (\id{Action}(a) \land \id{Action}(a_p) \land \id{Time}(t) \land \id{Prec}(a_p,a) \land \id{Do}(a,t)  \\
 & \qquad \lir \exists t_p\ (\id{Time}(t_p) \land t_p < t \land \id{Do}(a_p,t_p))).
\end{align*}
This sentence describes that some action $a$ with precondition $a_p$ can only be performed at time point $t$ if $a_p$ is performed at some earlier time point $t_p$. Let $\app{I}_2$ be a structure such that 
\[ \app{I}_2(Prec(d_0,d_1) \land \ldots \land Prec(d_{n-1},d_{n})) = \true. \]
$\app{I}_2$ indicates that there is a chain of $n$ actions that need to be performed before $d_n$. The propagation algorithm can derive for input $T_2$ and $\app{I}_2$ that $d_n$ can certainly not be performed before the $(n+1)$th time point.
\end{example}

The INF sentences in Example~\ref{ex:nonmin} illustrate that the presented transformation from general FO theories to INF sentences produces a non-minimal set of INF sentences. Logic programming techniques may be applied to reduce this set. For instance, if \emph{unfolding} is applied for predicate $\id{Aux}_1$, sentence~\eqref{eq:finp1-1} is replaced by the shorter sentence $\forall x\forall y\ (\id{MutExcl}(x,y) \land \id{Selected}(x) \lir \neg\id{Selected}(y))$, and~\eqref{eq:finp1} can be omitted. Similarly, unfolding can be applied to omit~\eqref{eq:finp2}--\eqref{eq:finp4} and replace~\eqref{eq:finp5} by $\forall m\ (\id{Module}(m) \land \exists c\ (\neg\id{Selected}(c) \land \id{In}(c,m)) \lir \neg \id{Selected}(m))$. Sentences in $\infs{T_1}$ of the form $\forall x\forall y\ (\varphi \lir \id{Aux}_1(x,y))$ can be omitted because they are \emph{subsumed} by~\eqref{eq:finp1}. Sentences of the form $\varphi \lir \top$ can be omitted because they are tautologies. Etc. It depends on the practical implementation of the propagation algorithm whether optimizing the set of INF sentences in this manner leads to a significant speed-up.

For sentences $\varphi$ of some specific form, it is easy to directly associate sets of INF sentences that are smaller than $\infs{\varphi}$ but produce the same propagation. For instance, to a clause $\forall \xxx\ (L_1 \lor \ldots \lor L_n)$, the set $\{ \forall \xxx\ (\neg L_1 \land \ldots \land \neg L_{i-1} \land \neg L_{i+1} \land \ldots \land \neg L_n \lir L_i) \mid 1 \leq i \leq n \}$ could be associated. For sentence~\eqref{form:cc} of Example~\ref{ex:stud}, this is the set
\begin{align}
 & \forall x\forall y\ (\id{MutExcl}(x,y) \land \id{Selected}(x) \lir \neg\id{Selected}(y)), \label{eq:redstart} \\
 & \forall x\forall y\ (\id{MutExcl}(x,y) \land \id{Selected}(y) \lir \neg\id{Selected}(x)), \\
 & \forall x\forall y\ (\id{Selected}(x) \land \id{Selected}(y) \lir \neg\id{MutExcl}(x,y)), \label{eq:redend}
\end{align}
instead of the fourteen sentences in $\infs{\eqref{form:cc}}$. It is noteworthy that extensively applying simplification techniques as described in the previous paragraph reduces $\infs{\eqref{form:cc}}$ to the three sentences~\eqref{eq:redstart}--\eqref{eq:redend}.

\subsection{Notes on Precision}

Because of Proposition~\ref{prop:opt}, the result $\app{J}$ of applying Algorithm~\ref{alg:prop} on input theory $T$ and structure $\app{I}$ is less precise than $\rfi{T}(\app{I})$. As we will show in Example~\ref{ex:mirek1}, there are cases where $\app{J}$ is strictly less precise than $\rfi{T}(\app{I})$. For applications like, e.g., configuration systems and approximate query answering (see Section~\ref{sec:appli}), it is an important question for which $T$ and $\app{I}$ this loss in precision occurs.  

The loss of precision in Algorithm~\ref{alg:prop} on an input theory $T$ compared to $\rfi{T}$, is in principle due to three factors:
\begin{enumerate}
	\item Instead of propagating the theory $T$ as a whole, Algorithm~\ref{alg:prop} considers propagators for individual sentences, and combines them in a refinement sequence. As Example~\ref{ex:rfitrsl} shows, this may lead to a loss in precision.
	\item The theory is translated to ENF. 
	\item Instead of applying the complete propagator $\rfi{\varphi}$ for an ENF sentence $\varphi$, the incomplete propagators $\infp{\psi}$ for INF sentences $\psi \in \infs{\varphi}$ are applied.
\end{enumerate}
The following theorem states that under some easy-to-verify conditions, the third factor does not contribute to the loss in precision. The theorem indicates that $\infps{\infs{\varphi}}$ is essentially the ``right'' set of propagators to approximate $\rfi{\varphi}$.

\begin{theorem}[\cite{phd/Wittocx10}]~\label{th:etoi}
If $\varphi$ is an ENF sentence such that no predicate occurs more than once in it, $\rfi{\varphi} = \iprop \circ \lim_{\infps{\infs{\varphi}}}$.
\end{theorem}

The inconsistency propagator $\iprop$ is needed in the theorem to cope with a small technical detail. The only strictly four-valued structure that can be obtained by applying a complete propagator is the most precise structure $\mapp$. This is in general not the case for the propagator $\lim_{\infps{\infs{\varphi}}}$. Applying the inconsistency propagator solves this technical detail.

Concerning the loss in precision due to the first two factors mentioned above, it is worth noting that predicate introduction may actually lead to more precise propagation. We illustrate this on an example.
\begin{example}\label{ex:mirek1}
Consider the propositional theory $T$ consisting of the two sentences $(P \lor Q)$ and $((P \lor Q) \lir R)$. Clearly, $R$ is true in every model of $T$ and therefore $\rfi{T}(\lapp)(R) = \true$. However, $\trsl{T}(\lapp) = \lapp$. Intuitively, this loss in precision is due to the fact that a three-valued structure cannot ``store'' the information that $(P \lor Q)$ is true in every model of $T$ if neither $P$ nor $Q$ is true in every model. However, if we apply predicate introduction to translate $T$ to the theory $T'$ consisting of the sentences
\[
\id{Aux} \ler P \lor Q, \qquad \id{Aux}, \qquad \id{Aux} \lir R,
\]
there is no loss in precision: $\trsl{T'}(\lapp)(R) = \true$. The fact that $(P \lor Q)$ must be true is ``stored'' in the interpretation of the introduced predicate $\id{Aux}$.
\end{example}

We refer to the work of \citeN{tods/DeneckerCBA10} for results on precision in the context of approximate query answering in incomplete databases under local closed world assumptions. It is a topic for future research to extend these results to our more general context.

\section{Symbolic Propagation}\label{sec:sip}

In this section, we discuss the symbolic version of INF propagators. To this end, we first introduce the notion of a \emph{symbolic structure}. Intuitively, a symbolic structure $\Phi$ relates a vocabulary $\symvoc$ to a vocabulary $\voc$. It does so by defining for every predicate of $\voc$ a query over $\symvoc$. This relationship can be used for mapping a structure over a vocabulary $\symvoc$ to a structure over $\voc$. This is reminiscent of materializing the intentional predicates in a deductive database. The relationship can also be used to map a formula over $\voc$ to a formula over $\symvoc$. This is reminiscent of reducing a query over the intensional predicates of a deductive database to a query over the extensional predicates. A symbolic structure is similar to an \emph{interpretation between theories}~\cite{Enderton01}, but it does not alter quantifiers.

Once symbolic structures are defined, symbolic INF propagators are introduced. These propagators map symbolic structures to symbolic structures, in a similar way as non-symbolic INF propagators do for non-symbolic structures. As we will explain, symbolic propagation is beneficial when precision is less important than efficiency, when only parts of the result of propagation are of interest, or when propagation for a fixed theory needs to be performed for several structures. 

In theory, the vocabularies $\symvoc$ and $\voc$ that are connected by a symbolic structure need not be related. However, in all practical applications we investigated so far, $\symvoc$ is a subset of $\voc \cup \tf{\voc}$. The interpretation of the predicates in $\symvoc$ acts as the input to non-symbolic propagation.

\subsection{Symbolic Structures}

\begin{definition}
A \emph{symbolic two-valued $\voc$-structure} $\Phi$ over $\symvoc$ consists of a query $P^{\Phi}$ for each predicate $P \in \voc$. The query $P^{\Phi}$ for a predicate $P/n$ is of the form $\{ (x_1,\ldots,x_n) \mid \varphi \}$ with $\varphi$ a formula over $\symvoc$.
\end{definition}
For the rest of this section, when we use the term {\em symbolic structure}, we mean a symbolic $\voc$-structure over $\symvoc$.  

A symbolic two-valued structure $\Phi$ can be used to map a two-valued $\symvoc$-structure $E$ with domain $D$ to a two-valued $\voc$-structure, denoted $\Phi(E)$, over the same domain $D$; it uses the queries to define the predicates of $\voc$ in $\Phi(E)$.

\begin{definition}
Let $E$ be a $\symvoc$-structure. Then $\Phi(E)$ denotes a $\voc$-structure which, for each predicate in $P \in \voc$, is defined as $P^{\Phi(E)} = (P^{\Phi})^E$.
\end{definition}

\begin{example}\label{ex:rhombus}
Let $\voc$ be the vocabulary $\{ \id{Rhombus}/1 \}$ and $\symvoc$ the vocabulary $\{ \id{Quadrilateral}/1, \id{EqualSides}/1 \}$. An example of a symbolic $\voc$-structure over $\symvoc$ is the symbolic structure $\Phi$ that assigns $\id{Rhombus}^{\Phi} = \{ x \mid \id{Quadrilateral}(x) \land \id{EqualSides}(x) \}$. If $E$ is the $\symvoc$-structure with domain $D = \{ a, b, c \}$ that assigns $\id{Quadrilateral}^E = \{ a, b\}$ and $\id{EqualSides}^E = \{ b,c\}$, then $\Phi(E)$ is the $\voc$-structure with domain $D$ that assigns $\id{Rhombus}^{\Phi(E)} = \{ x \mid \id{Quadrilateral}(x) \land \id{EqualSides}(x) \}^E = \{ b \}$.
\end{example}

Given $E$, $\Phi$ can be seen as a symbolic description of $\Phi(E)$. Given a set $V$ of $\symvoc$-structures, $\Phi$ can be seen as describing the set $\{ \Phi(E) \mid E \in V \}$ of $\voc$-structures. 

A symbolic structure $\Phi$ can also be used to map a formula over $\voc$ to a formula over $\symvoc$. It uses the queries to ``unfold'' the predicates in $\voc$.

\begin{definition}
Let $\varphi$ be a formula over $\voc$ and $\Phi$ a symbolic structure. Then $\Phi(\varphi)$ denotes the formula over $\symvoc$ obtained by replacing each occurrence of an atom $P(\yyy)$ in $\varphi$ by $\psi[\xxx/\yyy]$, where $P^{\Phi} = \{ \xxx \mid \psi \}$.
\end{definition}

The following proposition relates models of $\varphi$ with models of $\Phi(\varphi)$.

\begin{proposition}\label{lem:ev}
  For every formula $\varphi$ over $\voc$, symbolic structure $\Phi$, $\symvoc$-structure $E$ and variable assignment $\theta$, $(\Phi(E))\theta \models \varphi$ iff $E\theta \models \Phi(\varphi)$.
\end{proposition}

\begin{example}
Let $\voc$, $\symvoc$, $\Phi$, and $E$ be as in Example~\ref{ex:rhombus}, and let $\varphi$ be the sentence $\exists y\ \id{Rhombus}(y)$. Then $\Phi(\varphi)$ is the sentence $\exists y\ (\id{Quadrilateral}(y) \land \id{EqualSides}(y))$. Clearly, $E \models \Phi(\varphi)$ and $\Phi(E) \models \varphi$.
\end{example}

We call a symbolic $\tf{\voc}$-structure over $\symvoc$ a \emph{four-valued symbolic $\voc$-structure} over $\symvoc$. Such a structure $\app{\Phi}$ can be used to map a two-valued $\symvoc$-structure $E$ to a four-valued $\voc$-structure $\app{I}$, namely the structure such that $\tf{\app{I}} = \app{\Phi}(E)$. As such, it can be seen as a symbolic description of a class of four-valued structures over $\voc$. Abusing notation, we identify $\app{\Phi}(E)$ with $\app{I}$.

A four-valued symbolic structure $\app{\Phi}$ can also be used to map a formula $\varphi$ over $\voc$ to a pair of formulas over $\symvoc$, namely the pair of $\symvoc$-formulas $(\app{\Phi}(\ct{\varphi}),\app{\Phi}(\cf{\varphi}))$ which we denote as $\app{\Phi}(\varphi)$. Combining Proposition~\ref{lem:ev} and Proposition~\ref{prop:ftot} then yields the desired result that the $\app{\Phi}(\varphi)$ is a description of the truth value of $\varphi$ in the four-valued structures represented by $\app{\Phi}$. That is, for every $\symvoc$-structure $E$ and variable assignment $\theta$, $\app{\Phi}(E)\theta(\varphi) = E\theta(\app{\Phi}(\varphi))$. In other words, to evaluate $\varphi$ in structure $\app{\Phi}(E)$ and variable assignment $\theta$, one can first evaluate $\varphi$ symbolically in $\app{\Phi}$ and then in $E$ and $\theta$.

\begin{example}\label{ex:symstud}
  Let $\voc$ be $\{ \id{Module}/1, \id{Selected}/1, \id{In}/2, \id{MutExcl}/2 \}$ (the vocabulary from Example~\ref{ex:stud}) and let $\symvoc$ be $\{ \id{Module}/1, \id{In}/2, \id{MutExcl}/2, \ct{\id{Selected}}/1\}$. Let $\app{I}_0$ be the $\voc$-structure from Example~\ref{ex:stud} and let $E$ be the two-valued $\symvoc$-structure that assigns $\{ c_1 \}$ to $\ct{\id{Selected}}$ and corresponds to $\app{I}_0$ on the symbols of $\voc \cap \symvoc$. Define the four-valued symbolic $\voc$-structure $\app{\Phi}$ over $\symvoc$ by 
\begin{align*}
	& \ct{\id{In}}^{\app{\Phi}} = \{ (c,m) \mid \id{In}(c,m) \}		& & \ct{\id{MutExcl}}^{\app{\Phi}} = \{ (x,y) \mid \id{MutExcl}(x,y) \}  \\
	& \cf{\id{In}}^{\app{\Phi}} = \{ (c,m) \mid \neg\id{In}(c,m) \}	& & \cf{\id{MutExcl}}^{\app{\Phi}} = \{ (x,y) \mid \neg\id{MutExcl}(x,y) \} \\
	& \ct{\id{Module}}^{\app{\Phi}} = \{ m \mid \id{Module}(m) \}     & & \ct{\id{Selected}}^{\app{\Phi}} = \{ c \mid \ct{\id{Selected}}(c) \} \\
	& \cf{\id{Module}}^{\app{\Phi}} = \{ m \mid \neg\id{Module}(m) \} & & \cf{\id{Selected}}^{\app{\Phi}} = \{ c \mid \bot \}
\end{align*}
It can be checked that $E(\app{\Phi})$ corresponds to $\app{I}_0$.
Let $\varphi$ be the sentence \[\forall c\forall m\ (\neg\id{Selected}(m) \lor \neg \id{In}(m,c) \lor \id{Selected}(c)).\] Then $\ct{\varphi}$ and $\cf{\varphi}$ are given by, respectively,
\begin{align*}
	& \forall c\forall m\ (\cf{\id{Selected}}(m) \lor \cf{\id{In}}(m,c) \lor \ct{\id{Selected}}(c)), \\
	& \exists c\exists m\ (\ct{\id{Selected}}(m) \land \ct{\id{In}}(m,c) \land \cf{\id{Selected}}(c)).
\end{align*}
The evaluation of $\ct{\varphi}$ and $\cf{\varphi}$ in $\app{\Phi}$ are, respectively, the sentences
$\forall c\forall m\ (\bot \lor \neg\id{In}(c,m) \lor \ct{\id{Selected}}(c))$ and $\exists c\exists m\ (\ct{\id{Selected}}(m) \land \id{In}(c,m) \land \bot)$. These two sentences are false in $E$, and therefore $\varphi$ is unknown in $\app{I}_0$.  
\end{example}

\subsection{Symbolic Propagators}\label{ssec:symprops}

We now lift propagators to the symbolic level. 
\begin{definition}\label{def:symp}
A \emph{symbolic propagator $\prps$} for a theory $T$ is an operator on the set of four-valued symbolic structures over $\symvoc$ such that for each $\symvoc$-structure $E$ and symbolic structure $\app{\Phi}$, the following conditions are satisfied:
\begin{longitem}
	\item $\app{\Phi}(E) \leqp \prps(\app{\Phi})(E)$
	\item for every model $M$ of $T$ such that $\app{\Phi}(E) \leqp M$, also $\prps(\app{\Phi})(E) \leqp M$.
\end{longitem}
\end{definition}
Note that these two conditions on symbolic propagators are similar to the conditions on non-symbolic propagators. As is the case for non-symbolic propagators, the composition $\prps_2 \circ \prps_1$ of two symbolic propagators for theory $T$ is again a symbolic propagator for $T$. 

We say that a symbolic propagator $\prps$ \emph{describes} a non-symbolic propagator $\prp$ if for every symbolic four-valued structure $\app{\Phi}$ over $\symvoc$ and every $\symvoc$-structure $E$, $\prps(\app{\Phi})(E) = \app{\Phi}(\prp(E))$.  It is straightforward to check that if $\prps_1$ describes $\prp_1$ and $\prps_2$ describes $\prp_2$, then $\prps_2 \circ \prps_1$ describes $O_2 \circ O_1$. It follows that symbolic propagators can be used to describe finite refinement sequences. Indeed, let $V$ be a set of propagators such that for each $\prp \in V$, there exists a symbolic propagator $\prps$ describing $\prp$. Let $\langle \app{J}_i \rangle_{0 \leq i \leq n}$ be a $V$-refinement sequence from $\app{\Phi}(E)$ and denote by $\prp_i$ a propagator such that $\prp_i(\app{J}_i) = \app{J}_{i+1}$. Then $\app{J}_n = \prps_{n-1}(\ldots(  \prps_{0}(\app{\Phi}))\ldots)(E)$ where $S_i$ denotes the symbolic propagator that describes $O_i$ for $0 \leq i < n$. As such, $\prps_{n-1}(\ldots(  \prps_{0}(\app{\Phi}))\ldots)$ can be seen as a symbolic representation of the refinement sequence $\langle \app{J}_i \rangle_{0 \leq i \leq n}$. To describe transfinite refinement sequences with symbolic propagators, we would in general need symbolic $\voc$-structures that assign queries over an infinitary logic to the symbols of $\voc$.

We now introduce symbolic INF propagators. For the rest of this section, let $\app{\Phi}$ be a four-valued symbolic $\voc$-structure over $\symvoc$ and $E$ a $\symvoc$-structure. If two queries $\{ \xxx \mid \psi \}$ and $\{ \yyy \mid \chi \}$ have the same arity, i.e., $\size{\xxx} = \size{\yyy}$, we denote by $\{ \xxx \mid \psi \} \cup \{ \yyy \mid \chi \}$ the query $\{ \zzz \mid \psi[\xxx/\zzz] \lor \chi[\yyy/\zzz] \}$, where $\zzz$ is a tuple of new variables. Note that $(\{ \xxx \mid \psi \} \cup \{ \yyy \mid \chi \})^E = \{ \xxx \mid \psi \}^E \cup \{ \yyy \mid \chi \}^E $ for every structure $E$.
\begin{definition}
Let $\varphi$ be the INF sentence $\forall \xxx\ (\psi \lir P(\xxx))$. The \emph{symbolic INF propagator} $\sinfp{\varphi}$ is defined by 
\[
	Q^{\sinfp{\varphi}(\app{\Phi})} = 
	\begin{cases}
		\ct{P}^{\app{\Phi}} \cup  \{ \xxx \mid \app{\Phi}(\ct{\psi}) \}	& \text{if $Q = \ct{P}$} \\
		Q^{\app{\Phi}}														& \text{otherwise}. 
	\end{cases}
\]
That is, the queries for predicates different from $\ct{P}$ are copied from $\app{\Phi}$ and the query for $\ct{P}$ is extended with the mapping upon vocabulary $\symvoc$ of $\ct{\psi}$. If $\varphi$ is the INF sentence $\forall \xxx\ (\psi \lir \neg P(\xxx))$, then
$\sinfp{\varphi}$ is defined by
\[
	Q^{\sinfp{\varphi}(\app{\Phi})} = 
	\begin{cases}
		\cf{P}^{\app{\Phi}} \cup \{ \xxx \mid \app{\Phi}(\ct{\psi}) \} )	& \text{if $Q = \cf{P}$} \\
		Q^{\app{\Phi}}														& \text{otherwise}. 
	\end{cases}
\]
\end{definition}
The following result states the desired property that symbolic INF propagators are the symbolic counterpart of non-symbolic INF propagators.
\begin{proposition}\label{prop:sinfp}
$\sinfp{\varphi}$ describes $\infp{\varphi}$ for every INF sentence $\varphi$.
\end{proposition}

Proposition~\ref{prop:sinfp} implies that one can execute the propagation algorithm (Algorithm~\ref{alg:prop}) using symbolic INF propagators in step~\ref{it:propalgo} instead of non-symbolic ones. We refer to this symbolic version of the algorithm as the \emph{symbolic propagation algorithm}. 

\begin{example}
Consider the following INF sentences\footnote{These sentences are some
  of the INF sentences obtained when reducing a standard \foid
  encoding of the Hamiltonian path problem to INF. The predicate
  $\id{InHam}$ represents the edges in the Hamiltonian path. The predicate \id{Aux} is an auxiliary predicate, introduced when the sentence $\forall x\forall y\forall z\ (\id{InHam}(x,y) \land \id{InHam}(x,z) \lir y = z)$ --- which states that the path does not split --- is reduced to ENF.}:
\begin{align}
	& \forall x\forall y\ (\neg \id{Edge}(x,y) \lir \neg \id{InHam}(x,y)) \label{eq:symone} \\ 
	& \forall x\forall y\ (\id{Start}(y) \lir \neg \id{InHam}(x,y)) \label{eq:symtwo} \\
	& \forall x\forall y\forall z\ (\neg \id{InHam}(x,y) \land \neg \id{InHam}(x,z) \lir \id{Aux}(x,y,z)). \label{eq:symthree}
\end{align}
Let $\symvoc$ be the vocabulary $\{\id{Edge}/2,\id{Start}/1\}$ and let $\app{\Phi}_0$ be the symbolic structure over $\symvoc$ assigning
\begin{align*}
& \ct{\id{Edge}}^{\app{\Phi}_0} = \{ (x,y) \mid \id{Edge}(x,y) \}	& & \cf{\id{Edge}}^{\app{\Phi}_0} = \{ (x,y) \mid \neg\id{Edge}(x,y) \} \\
& \ct{\id{Start}}^{\app{\Phi}_0} = \{ x \mid \id{Start}(x) \}	& & \cf{\id{Start}}^{\app{\Phi}_0} = \{ x \mid \neg\id{Start}(x) \} \\
& \ct{\id{InHam}}^{\app{\Phi}_0} = \{ (x,y) \mid \bot \}	& & \cf{\id{InHam}}^{\app{\Phi}_0} = \{ (x,y) \mid \bot \} \\
& \ct{\id{Aux}}^{\app{\Phi}_0} = \{ (x,y,z) \mid \bot \}	& & \cf{\id{Aux}}^{\app{\Phi}_0} = \{ (x,y,z) \mid \bot \}.
\end{align*}
Applying $\sinfp{\eqref{eq:symone}}$ on $\app{\Phi}_0$ yields a symbolic structure $\app{\Phi}_1$ that assigns $\{ (x,y) \mid \bot \lor \neg \id{Edge}(x,y)\}$ to $\cf{\id{InHam}}$. Applying $\sinfp{\eqref{eq:symtwo}}$ on $\app{\Phi}_1$ produces symbolic structure $\app{\Phi}_2$ assigning $\{ (x,y) \mid \bot \lor \neg \id{Edge}(x,y) \lor \id{Start}(y) \} $ to $\cf{\id{InHam}}$. Finally, the result of applying $\sinfp{\eqref{eq:symthree}}$ to $\app{\Phi}_2$ assigns 
\begin{equation}
\{ (x,y,z) \mid \bot \lor ((\bot \lor \neg \id{Edge}(x,y) \lor \id{Start}(y)) \land (\bot \lor \neg \id{Edge}(x,z) \lor \id{Start}(z))) \} \label{eq:complex}
\end{equation}
to $\ct{\id{Aux}}$.
\end{example}

Observe that computing $\sinfp{\varphi}(\app{\Phi})$ takes time
$\mathcal{O}(\size{\varphi} \cdot \size{\app{\Phi}})$, while computing
$\infp{\varphi}(\app{I})$ takes time
$\mathcal{O}(\size{\app{I}}^{\size{\varphi}})$ since evaluating a
formula $\varphi$ in a structure $I$ takes time
$\mathcal{O}(\size{I}^{\size{\varphi}})$ \cite{GradelKLMSVVW07}. This
indicates a possible benefit of using symbolic INF propagators instead
of non-symbolic ones. However, this gain in efficiency does not come
for free. One problem is that testing whether a sequence of symbolic
structures is stabilizing is undecidable, because it boils down to testing logical equivalence of FO formulas.
Another problem concerning symbolic refinement sequences is the size of symbolic structures. The size of $\sinfp{\varphi}(\app{\Phi})$ is $\mathcal{O}(\size{\varphi} \cdot \size{\app{\Phi}})$. As such, the size of the last element of a refinement sequence constructed using symbolic INF propagators is exponential in the length of the sequence, while for non-symbolic refinement sequences from a finite structure, the size of the last element is polynomial in the size of that structure. The exponential growth of the symbolic structures can sometimes, but not always, be avoided by replacing the queries assigned by a structure by equivalent, but smaller queries. For example,~\eqref{eq:complex} could be replaced by the shorter, equivalent query
\[ \{ (x,y,z) \mid (\neg \id{Edge}(x,y) \lor \id{Start}(y)) \land (\neg \id{Edge}(x,z) \lor \id{Start}(z)) \}. \]
\citeN{jair/WittocxMD10} describe a detailed implementation of the symbolic propagation algorithm using first-order binary decision diagrams \cite{Goubault95}.

We expect symbolic propagation to be useful in applications where
precision is less important than efficiency, and where the evaluation
$\app{\Phi}(E)$ of the last structure $\app{\Phi}$ of a refinement sequence in $\symvoc$-structure $E$ need not be computed completely. Grounding (Section~\ref{ssec:grounding}) and approximate query answering (Section~\ref{ssec:aqa}) are two examples of such applications. 

\section{Propagation for FO(ID)}\label{sec:foid}

One of the famous examples of concepts that are not expressible in FO is the concept of reachability in a graph. In fact, most concepts that require a recursive definition cannot be expressed in FO. Nevertheless, inductive definitions appear in many real-life computational problems such as automated planning or problems involving dynamic systems \cite{DeneckerT07,tocl/DeneckerT08}. In this section, we extend the propagation algorithm to an extension of FO with inductive, i.e., recursive, definitions. 

\subsection{Inductive Definitions}

Like a rule set, a \emph{definition} $\Delta$ is a finite set of rules of the form $\forall \xxx \ (P(\xxx) \rul \varphi[\yyy])$.  Predicates that appear in the head of a rule of $\Delta$ are called \emph{defined predicates} of $\Delta$. The set of all defined predicates of $\Delta$ is denoted $\defi{\Delta}$. All other symbols are called \emph{open} with respect to $\Delta$. The set of all open symbols of $\Delta$ is denoted $\open{\Delta}$. 

\begin{example}
The following definition defines the predicate $\id{Reach}$ in terms of open predicate $\id{Edge}$. 
\begin{equation}\label{eq:reachdef}
\left\{ 
	\begin{array}{rl}
		\forall x\forall y\ (\id{Reach}(x,y) & \rul \id{Edge}(x,y)), \\
		\forall x\forall y\ (\id{Reach}(x,y) & \rul \exists z\ (\id{Reach}(x,z) \land \id{Reach}(z,y)))
	\end{array}
\right\}
\end{equation}
Informally, this definition expresses that $y$ can be reached from $x$ in the graph represented by $\id{Edge}$, if either there is an edge between $x$ and $y$, i.e., $\id{Edge}(x,y)$ is true, or if there is some intermediate node $z$ such that $z$ can be reached from $x$ and $y$ can be reached from $z$. 
\end{example}

The formal semantics of definitions is given by their well-founded model \cite{GelderRS91}.  We borrow the presentation of this semantics from \citeN{lpnmr/DeneckerV07}. 
\begin{definition}\label{def:wfi}
Let $\Delta$ be a definition and $\app{I}$ a finite three-valued structure. A \emph{well-founded induction for $\Delta$ extending $\app{I}$} is a (possibly transfinite) sequence $\langle \app{J}_{\xi} \rangle_{0 \leq \xi \leq \alpha}$ of three-valued structures such that 
\begin{enumerate}
\item $\app{J}_0 = \res{\app{I}}{\open{\Delta}} + \lapp_{\defi{\Delta}}$;
\item $\app{J}_{\lambda} = \lub_{\leqp}(\{ \app{J}_{\xi} \mid \xi < \lambda\})$ for every limit ordinal $\lambda \leq \alpha$;
\item For every ordinal $\xi$, $\app{J}_{\xi+1}$ relates to $\app{J}_{\xi}$ in one of the following ways:
\begin{enumerate}
\item $\app{J}_{\xi+1} = \app{J}_{\xi}[V/\true]$, where $V$ is a set of domain atoms such that for each $P(\ddd) \in V$, $P^{\app{J}_{\xi}}(\ddd) = \unkn$ and there exists a rule $\forall \xxx\ (P(\xxx) \rul \varphi)$ in $\Delta$ such that $\app{J}_{\xi}[\xxx/\ddd](\varphi) = \true$. \label{it:stept}
\item $\app{J}_{\xi+1} = \app{J}_{\xi}[U/\false]$, where $U$ is a set of domain atoms, such that for each $P(\ddd) \in U$, $P^{\app{J}_{\xi}}(\ddd) = \unkn$ and for all rules $\forall \xxx\ (P(\ttt) \rul \varphi)$ in $\Delta$, $\app{J}_{\xi+1}[\xxx/\ddd](\varphi) = \false$. \label{it:stepf}
\end{enumerate}
\end{enumerate}
\end{definition}
Intuitively,~\eqref{it:stept} says that domain atoms $P(\ddd)$ can be made true if there is a rule with $P(\xxx)$ in the head and body $\varphi$ such that $\varphi$ is already true, given a variable assignment that interprets $\xxx$ by $\ddd$. On the other hand~\eqref{it:stepf} explains that $P(\ddd)$ can be made false if there is no possibility of making a corresponding body true, except by circular reasoning. The set $U$, called an \emph{unfounded set}, is a witness to this: making all atoms in $U$ false also makes all corresponding bodies false.

A well-founded induction is called \emph{terminal} if it cannot be extended anymore. The limit of a terminal well-founded induction is its last element. \citeN{lpnmr/DeneckerV07} show that each terminal well-founded induction for $\Delta$ extending $\app{I}$ has the same limit, which corresponds to the well-founded model of $\Delta$ extending $\res{\app{I}}{\open{\Delta}}$. The well-founded model is denoted by $\wfm{\Delta}{\app{I}}$. In general, $\wfm{\Delta}{\app{I}}$ is three-valued.

A two-valued structure $I$ satisfies definition $\Delta$, denoted $I
\models \Delta$, if $I = \wfm{\Delta}{I}$. The extension of FO with
inductive definition
is called \foid. A \foid theory is a set of FO sentences and definitions. A two-valued structure satisfies an \foid theory $T$ if it satisfies every sentence and every definition of $T$.

The \emph{completion} of a definition $\Delta$ is an \foid theory that is weaker than $\Delta$:
\begin{definition}
The \emph{completion}\index{completion} of a definition $\Delta$ is the FO theory that contains for every $P \in \defi{\Delta}$ the sentence
\[\forall \xxx\ (P(\xxx) \ler ( \exists \yyy_1 (\xxx = \yyy_1 \land \varphi_1) \lor \ldots \lor \exists \yyy_n(\xxx = \yyy_n \land \varphi_n) )), \]
where $\forall \yyy_1\ (P(\yyy_1) \rul \varphi_1)$, \ldots, $\forall \yyy_n\ (P(\yyy_n) \rul \varphi_n)$ are the rules in $\Delta$ with $P$ in the head. 
\end{definition}
We denote the completion of $\Delta$ by $\comp{\Delta}$. If $T$ is an
\foid theory then we denote by $\comp{T}$ the result of replacing in
$T$ all definitions by their completion. The following result states
that the completion of $T$ is weaker than $T$.  

\begin{theorem}[\cite{tocl/DeneckerT08}]\label{th:comp} $\Delta
  \models \comp{\Delta}$ and $T \models \comp{T}$ for every definition
  $\Delta$ and \foid theory $T$.
\end{theorem}

\subsection{Propagation for Definitions}

In this section, we consider two approaches to extend the propagation
method to \foid. First, we discuss the application of our propagation method
  for FO on the completion of \foid theories. Secondly, we define an INF propagator for definitions. 

\subsubsection{Propagation on the completion}

It follows from Theorem~\ref{th:comp} that the propagators obtained by applying Algorithm~\ref{alg:prop} on $\comp{T}$ are propagators for the theory $T$. However, note that a complete propagator for $\comp{T}$ can be incomplete for $T$. For example, consider the definition $\Delta := \{ P \rul P \}$. This definition has only one model, in which $P$ is false. Hence, $(\rfi{\Delta}(\lapp))(P) = \false$. The completion of $\Delta$ is the sentence $(P \ler P)$, which has a model making $P$ true and one making $P$ is false. Therefore $(\rfi{\comp{\Delta}}(\lapp))(P) = \unkn$. We conclude that $\rfi{\comp{\Delta}} \neq \rfi{\Delta}$. Moreover, $\rfi{\comp{\Delta}}$ is not inducing for $\Delta$, that is, it may not recognize that a given two-valued structure is not a model of $\Delta$.  

If $\app{I}$ is a finite structure and $T$ a \foid theory, there exists an FO theory $T'$ such that the models of $T$ approximated by $\app{I}$ are precisely the models of $T'$ approximated by $\app{I}$. Such a theory can be constructed by applying propositionalization (see, e.g.,~\cite{jair/WittocxMD10}), followed by the transformations described by \citeN{ecai/Janhunen04} or by \citeN{iclp/PelovT05}. Propagation on $T$ and $\app{I}$ can then be obtained by applying propagation on $T'$ and $\app{I}$. The benefit of this approach is a gain in precision. In particular, the resulting propagator is inducing. On the other hand, $T'$ can be exponentially larger than $T$, which has repercussions on the efficiency of (symbolic) propagation.

\subsubsection{Propagators for definitions}

A second way to extend our propagation method to \foid is to introduce special purpose propagators involving definitions. 

\begin{definition}
The propagator $\infp{\Delta}$ for a definition $\Delta$ is defined by 
\[ P^{\infp{\Delta}(\app{I})}(\ddd) = 
	\begin{cases}
		\lub_{\leqp}\{ \true, P^{\app{I}} \} & \text{if $P^{\wfm{\Delta}{\app{I}}}(\ddd) = \true$} \\
		\lub_{\leqp}\{ \false,P^{\app{I}} \} & \text{if $P^{\wfm{\Delta}{\app{I}}}(\ddd) = \false$} \\
		P^{\app{I}}									 & \text{otherwise.}
	\end{cases}
\]
\end{definition}

It follows from the definition of well-founded induction that
$\infp{\Delta}$ is a monotone propagator for every definition
$\Delta$. Moreover, for finite structures $\app{I}$,
$\wfm{\Delta}{\app{I}}$ can be computed in polynomial time in
$\size{\app{I}}$. As such, $\infp{\Delta}$ is a propagator with
polynomial-time data complexity. Note that this propagator only
propagates information from the body of the definition to the head; to
propagate from head to body, one needs propagators derived from the
completion.

It is an open question whether the propagator $\lim_{\infps{V}}$ can be represented by a (positive or monotone) rule set if $V$ may contain both INF sentences and definitions. Results from fixpoint logics (see, e.g.,~\cite{GradelKLMSVVW07}) suggest that this will be possible when only finite structures are considered, but impossible in general. We expect that even if it is possible to represent $\lim_{\infps{V}}$ by a rule set $\Delta_V$, this result will not be of practical importance, since $\Delta_V$ will be rather complicated. The same remark applies for symbolic propagators simulating $\lim_{\infps{V}}$.

Recently, \citeN{jelia/VlaeminckWVDB10} showed how to represent the propagator $\lim_{\infps{V}}$ by a \emph{nested fixpoint definition} \cite{phd/Hou10}. Methods to evaluate such nested definitions are currently being investigated \cite{tplp/HouDD10}. The extent to which the theoretical and practical results about monotone rule sets can be adapted to nested fixpoint definitions will determine the usefulness of representing $\lim_{\infps{V}}$ by such definitions.

\section{Applications}\label{sec:appli}
In this section, we briefly sketch several applications of constraint propagation, namely finite model generation, improved grounding, approximate solving of universal second-order logic ($\forall$SO) model generation problems, declarative programming of configuration systems, and approximate query answering in incomplete databases. We refer to papers where these applications are discussed in more detail.

\subsection{Solving Constraint Satisfaction Problems}\label{ssec:finm}

The obvious application of constraint propagation is to solve CSPs. Many real-life computational problems are naturally cast as CSPs. Well-known examples are scheduling, planning, diagnosis, and lightweight dynamic system verification. A standard algorithm to solve a constraint satisfaction problem $\langle \mathcal{C}, V,\domf \rangle$ is to combine propagators for $C$ with a backtracking algorithm and a branching heuristic.

A \emph{model expansion} (MX) is the problem of finding for a given theory $T$ in some logic $\mcl$ and a structure $\app{I}$, a model of $T$ that is approximated by $\app{I}$. We denote MX for input theories in logic $\mcl$ by MX($\mcl$). As mentioned in Section~\ref{ssec:mgtocsp} any CSP can be mapped to a model expansion problem and vice versa. Often, the representation of a CSP by a model expansion problem in a suitable logic is compact and highly declarative (see, e.g., the encodings of problems used in the second ASP Competition~\cite{lpnmr/DeneckerVBGT09}). Similarly as for solving CSPs, MX(\foid) problems can be solved by combining our propagation method for \foid with a backtracking algorithm and branching heuristics. 

Most current MX (and ASP) solvers take another approach. First, they reduce the input theory and domain to an equivalent propositional theory $T_g$. This process is called \emph{grounding}. Next, an (extended) SAT solver is applied to efficiently find a model for $T_g$. If found, this model then corresponds to a solution of the original problem. The benefit of this approach is that current SAT solvers are highly optimized. On the other hand, the grounding phase is often a bottleneck since in general, it takes exponential time in the quantifier rank, i.e., the nesting depth of quantifiers, of the input theory. Consequently, there is a trade-off between applying fast unit propagation but first having to ground the input theory, and applying our slower propagation method but avoiding the grounding.

A potential future approach to avoid the trade-off is by combining both propagation methods, in a way similar to DPLL(T) \cite{jacm/NieuwenhuisOT06}. In this combined approach, some sentences of the theory are grounded to propositional theory $T_g$, the others --- preferably those with a large quantifier rank but low \emph{activity}, i.e., yielding few propagations --- are transformed to a set of INF sentences $V$. Next, a SAT solver is applied on $T_g$. Whenever the SAT solver derives that a certain propositional literal $L$ is true, and $L$ unifies with a literal in the condition of an INF $\varphi$ sentence in $V$, $\infp{\varphi}$ can be applied to derive the truth of other literals. In turn, these literals can be communicated back to the SAT solver.

\subsection{Improved Grounding}\label{ssec:grounding}

Our \foid propagation can be applied to improve current SAT based MX solvers more directly by improving the grounding size and time \cite{jair/WittocxMD10}.  In an MX problem with input theory $T$ and input structure $\app{I}$, $\app{I}$ is often used to encode input \emph{data} for the problem. For example, $\app{I}$ might contain the input graph for a graph colouring problem. The grounders for MX and ASP primarily reduce the size of the computed propositional theory and improve the grounding speed by cleverly omitting formulas that do not evaluate to $\unkn$ in $\app{I}$. It follows that grounding improves if one first applies propagation on $T$ and $\app{I}$ to obtain a more precise structure $\app{J}$. 

In the case of improving grounding, efficiency of propagation is more important than completeness. Indeed, detailed propagation will be performed afterwards by the SAT solver. For this reason, when implementing propagation to optimize the grounder \gidl \cite{jair/WittocxMD10}, we opted for the symbolic propagation algorithm. Experiments with \gidl show that the time taken by symbolic propagation is negligible compared to the overall grounding time, while on average, it reduces grounding size and time by 30\%, respectively 40\%. In some cases, symbolic propagation makes the difference between being able to ground a theory in less than 20 seconds, compared to not being able to ground it at all. As far as we know, no one thoroughly evaluated whether concrete propagation is suitable to improve grounding as well.

\subsection{Approximate Solving of $\forall$SO Model Expansion Problems}

As \citeN{MitchellT05} show, it is a direct consequence of Fagin's \citeyear{fagin74} seminal result that every MX(\foid) problem with a fixed \foid input theory and variable finite input structures is in \np. If the input theory is in universal second-order logic ($\forall$SO), MX problems are in $\Sigma^P_2$, and some of these problems are $\Sigma^P_2$-hard. A class of interesting problems that are naturally cast as MX($\forall$SO) problems are \emph{conformant planning problems}. A conformant planning problem is a planning problem where only partial information $\app{I}$ about the initial state is given. A solution is a fixed plan that is guaranteed to lead from any initial state approximated by $\app{I}$ to the desired goal state. There exists conformant planning problems where determining whether a conformant plan of length less than a given length $l$ exists is $\Sigma^P_2$-hard, even if $l$ is polynomial in the size of the problem.

\citeN{jelia/VlaeminckWVDB10} show how to approximate an MX($\forall$SO) problem by an MX(\foid) problem, in the sense that solutions of the latter are solutions of the former (but not necessarily vice versa). The representation of FO propagation by a rule set is crucial in the approximation.

\subsection{Configuration Systems}\label{ssec:config}

The application presented in the introduction is an example of a \emph{configuration system}. A configuration system helps a user to fill out a form in accordance with certain constraints. As noted by \citeN{ppdp/VlaeminckVD09}, due to the large amount of background knowledge involved, developing and maintaining a configuration system can be difficult when using (only) a traditional imperative programming method. Instead, encoding the background knowledge, e.g., the constraints, in logic and applying suitable automated reasoning methods may make these tasks much easier.

One of the tasks of a configuration system is to prevent the user from making invalid choices by automatically disabling such choices. For example, if courses $c_1$ and $c_2$ are mutually exclusive and a student selects the course $c_1$, the system described in the introduction should make selecting $c_2$ impossible. Using constraint propagation, this functionality  can be implemented in a declarative way: the constraints describing valid configurations are represented by a theory $T$, the current selection by a three-valued structure $\app{I}$. Then, propagation is applied to derive a more precise structure $\app{J}$. Each possible choice that is true according to $\app{J}$ is selected automatically by the system, each choice that is false is disabled. 

There are two main, albeit somewhat contradictory requirements for the propagation in this case. First, since a configuration system is interactive, the propagation should be efficient in order to respond sufficiently fast. Secondly, in an ideal system, the user can never make an invalid choice. To this end, the propagation should implement the complete propagator $\rfi{T}$. Indeed, if $\app{J} = \rfi{T}(\app{I})$ and a choice $P(\ddd)$ is unknown in $\app{J}$, then there exists a model of $T$, i.e., a valid configuration, where $P(\ddd)$ is true, and one where $P(\ddd)$ is false. As such, neither selecting nor deselecting $P(\ddd)$ is an invalid choice since in both cases a valid configuration remains reachable. The combination of both requirements shows the importance of investigating the precision of efficient propagators.

We refer to the work of \citeN{ppdp/VlaeminckVD09} for a more elaborated investigation of knowledge based configuration software and a discussion of related work. The approach to build configuration systems using propagation for \foid was implemented in a Java\texttrademark library \cite{msc/Calus11}. The library is available from \url{http://dtai.cs.kuleuven.be/krr/software/download}. Configuration systems built with this library turn out to be sufficiently fast and precise.

\subsection{Approximate Query Answering}\label{ssec:aqa}

A recent trend in databases is the development of approximate methods to reason about databases with incomplete knowledge. The incompleteness of the database may stem from the use of null values, or of a restricted form of closed world assumption \cite{aaai/Cortes-CalabuigDAB07}, or it arises from integrating a collection of local databases each based on its own \emph{local schema} into one virtual database over a \emph{global schema} \cite{icdt/GrahneM99}. In all these cases, the data complexity of certain and possible query answering is computationally hard (\conp, respectively \np). For this reason fast (and often very precise) polynomial approximate query answering methods have been developed, which compute an underestimation of the certain, and an overestimation of the possible answers.

The tables of an incomplete database are naturally represented as a
three-valued structure $\app{I}$. The integrity constraints, local
closed world assumption or mediator scheme corresponds to a logic
theory $T$. Answering a query $\{ \xxx \mid \varphi\}$ boils down to
computing the set of tuples $\ddd$ such that $M[\xxx/\ddd] \models
\varphi$ in every model $M$ of $T$ approximated by $\app{I}$ (certain
answers) and the set of tuples $\ddd$ such that $M[\xxx/\ddd] \models
\varphi$ for at least one $M \models T$ approximated by $\app{I}$
(possible answers). These sets can be approximated by $\{ \ddd \mid \app{J}[\xxx/\ddd](\varphi) = \true \}$, respectively $\{ \ddd \mid \app{J}[\xxx/\ddd](\varphi) \neq \false \}$, where $\app{J}$ is obtained by applying constraint propagation for $T$ on $\app{I}$. If a constraint propagation method with polynomial-time data complexity is used to compute $\app{J}$, computing the approximate query answers above also requires polynomial time in the size of the database. Of course, the more precise $\app{J}$ is, the more precise the obtained answers to the query are.

Approximate query answering is an application where symbolic propagation is important. There are several reasons why it is to be preferred above non-symbolic propagation. First of all, the size of real-life databases makes the application of non-symbolic propagation often too slow in practice, since it requires the storage of large intermediate tables. More importantly, each time the data is changed, the propagation needs to be repeated. This is not the case for the symbolic propagation, because symbolic propagation is independent of the data. Thirdly, symbolic propagation can be used for query rewriting. Indeed, given a symbolic structure $\app{\Phi}$, computed by propagation, an evaluation structure $E$ and a query $\varphi$, the approximation to the certain answers for $\varphi$ are given by the set $\{ \ddd \mid \app{\Phi}(E)[\xxx/\ddd](\varphi) = \true \}$.  This set is equal to $\{ \xxx \mid \ct{(\app{\Phi}(\varphi))} \}^E$.  Hence the query $\{ \xxx \mid \varphi\}$ can be rewritten to a new query $\{ \xxx \mid \ct{(\app{\Phi}(\varphi))}\}$, which is then evaluated in the database $E$. Next, one can use the various optimization strategies in current database management systems to efficiently compute the answers to the new query. Possible answers to $\varphi$ are obtained in a similar way.

Applying the non-symbolic version of Algorithm~\ref{alg:prop} for approximate query answering generalizes the algorithm of~\citeN{lpar/CortesDAB06}. Applying the symbolic version and rewriting the query generalizes the query rewriting technique presented by~\citeN{aaai/Cortes-CalabuigDAB07}. Conditions that ensure the answers to queries obtained via these methods are optimally precise, i.e., conditions that ensure completeness of the propagation algorithm in the context of incomplete, locally closed databases, were investigated by \citeN{tods/DeneckerCBA10}.

\section{Conclusions}

In this paper we presented constraint propagation as a basic form of
inference for FO theories. We introduced a general notion of
constraint propagators and briefly discussed the complete propagator
for a theory. Due to its high computational complexity, the complete
propagator cannot be applied in most real-life applications. Therefore
we investigated incomplete propagators, called INF propagators. These
propagators generalize the propagators for propositional logic
presented by~\citeN{aaai/McAllester90} and \citeN{compulog/Apt99}. 
Similar propagators were
proposed in the context of locally closed databases, where approximative query
answering in polynomial time was studied in a series of
papers~\cite{lpar/CortesDAB06,aaai/Cortes-CalabuigDAB07,kr/CortesDAB08}
culminating in~\cite{tods/DeneckerCBA10}. A first version of INF
propagators for full FO was presented in the context of
grounding~\cite{aaai/WittocxMD08}. Later we improved the propagators and
presented them in a more general context~\cite{WittocxMD2008:KR}. The
link with constraint programming is new in the current paper.
Besides their lower computational complexity, INF propagators for FO have other interesting properties: propagation using INF propagators can be represented by a monotone rule set and can be executed in a symbolic manner. The former property allows us to use existing systems and extensively studied methods to make efficient implementations of propagation. The latter property is important in contexts where data changes regularly or where only part of the results obtained by propagation is needed. 

We extended the results about propagation using INF propagators to the
logic \foid that extends FO with inductive definitions. Whether the
results about representation by a monotone definition or symbolic
propagation carry over to inductive definitions, is an open
question. 
Further transfer of techniques developed in the constraint programming
community to improve propagation for FO, is also an interesting direction
for future work.

FO and \foid can also be extended with aggregates~\cite{tplp/PelovDB07}. In many cases, the use of aggregates yields more succinct theories~\cite{ai/SimonsNS02}, and often faster reasoning~\cite{tplp/FaberPLDI08}. Aggregates appear in many real-life applications. The extension of our propagation method to aggregates is described in Appendix~\ref{sec:aggr}.

Finally, we discussed several applications that rely on constraint propagation as basic form of inference. 

\newpage

\appendix

\section{Aggregates}\label{sec:aggr}

Aggregates are (partial) functions that have a set as argument. An example is the function $\card$ returning the cardinality of a set. In many cases, the use of aggregates yields more succinct theories~\cite{ai/SimonsNS02}, and often faster reasoning~\cite{tplp/FaberPLDI08}. Aggregates appear in many real-life applications.

In this section, given a domain $D$ we assume that there are two types of variables: variables that can take values among all elements of $D$, and variables that only take values among the elements in $D$ that are real numbers. That is, we assume a restricted form of many-sorted logic. Furthermore, we assume that formulas and structures are well-typed, in the sense that terms occurring at a position where only a number can sensibly occur, evaluate in every structure to a number. For instance, in a term $x+1$, $x$ should be a variable only ranging over real numbers. \citeN{ijcai/TernovskaM09} and \citeN{phd/Wittocx10} provide more detailed descriptions about including arithmetics in FO.

\subsection{FO with Aggregates}

We denote the extension of FO with aggregates by \foagg. A \emph{set expression} in \foagg is an expression of the form $\{ \xxx \mid \varphi \}$, where $\xxx$ is a tuple of variables and $\varphi$ a formula.\footnote{The only difference between queries and set expressions is that a formula $\varphi$ in set expression $\{ \xxx \mid \varphi \}$ may contain free variables that are not among $\xxx$.} The value of set expression $\{ \xxx \mid \varphi \}$ in structure $I$ under variable assignment $\theta$ is denoted by $I\theta(\{ \xxx \mid \varphi\})$ and defined by $\{ \ddd \mid I\theta[\xxx/\ddd] \models \varphi \}$. A set $V$ of tuples of domain elements is \emph{numeric} if for each $\ddd \in V$, the first element of $\ddd$ is a real number. 

In this paper, we consider the aggregate function symbols $\card$, $\asum$, $\aprod$, $\amin$ and $\amax$. An \emph{aggregate term} is an expression of the form $\card(V)$, $\asum(V)$, $\aprod(V)$, $\amin(V)$ or $\amax(V)$, where $V$ is a set expression. An \emph{aggregate atom} is a formula of the form $x \leq \agg(V)$ or $x \geq \agg(V)$, where $x$ is a variable and $\agg(V)$ an aggregate term. An \foagg formula is defined like an FO formula, except that atoms may be FO atoms as well as aggregate atoms.\footnote{One can generalize \foagg by allowing aggregate terms in every position where an FO term is allowed. We use the restricted version here to facilitate the presentation. There exists an equivalence preserving transformation from the more general version to the restricted one.} We use formulas of the form $t \leq \agg(V)$ and $t < \agg(V)$, where $t$ is a term and $\agg(V)$ an aggregate term as shorthands for the formulas $\exists x\ (t=x \land x \leq \agg(V))$, respectively $\exists x\exists y\ (t=x \land x < y \land y \leq \agg(V))$. Similarly for formulas of the form $t \geq \agg(V)$ and $t > \agg(V)$.

For a set of tuples $V$, we define $\card(V)$ to be the cardinality of $V$. If $V$ is numeric, we define:
\begin{longitem}
	\item $\asum(V) = 0$ if $V=\emptyset$ and $\asum(V) = \sum_{(a_1,\ldots,a_n) \in V}(a_1)$ otherwise;
	\item $\aprod(V) = 1$ if $V=\emptyset$ and $\aprod(V) = \prod_{(a_1,\ldots,a_n) \in V}(a_1)$ otherwise;
	\item $\amin(V) = +\infty$ if $V=\emptyset$ and $\amin(V) = \min\{a_1 \mid (a_1,\ldots,a_n) \in V\}$ otherwise;
	\item $\amax(V) = -\infty$ if $V=\emptyset$ and $\amax(V) = \max\{a_1 \mid (a_1,\ldots,a_n) \in V\}$ otherwise.
\end{longitem}
Let $I$ be a finite structure with domain $D$ and $\theta$ a variable assignment. The satisfaction relation for \foagg is defined by adding the following base cases to the satisfaction relation for FO: 
\begin{longitem}
\item $I\theta \models x \leq \agg(V)$ if
  $\theta(x)$ is a real number and $\theta(x) \leq \agg(I\theta(V))$;
\item $I\theta \models x \geq \agg(V)$ if
  $\theta(x)$ is a real number and $\theta(x) \geq \agg(I\theta(V))$.
\end{longitem}

To define the value of an aggregate atom in a three-valued structure, we first introduce three-valued sets. A \emph{three-valued set} is a set where each element is annotated with one of the truth values $\true,\false$ or $\unkn$. We denote the annotations by superscripts. A three-valued set $\app{V}$ \emph{approximates} a class of sets, namely all sets that certainly contain the elements of $V$ annotated by $\true$ and possibly some of the elements of $V$ annotated by $\unkn$. For example, $\{ a^{\true}, b^{\false}, c^{\unkn} \}$ denotes a three-valued set, approximating the sets $\{ a \}$ and $\{ a,c \}$. 

In a three-valued structure $\app{I}$, a set expression $\{ \xxx \mid \varphi \}$ evaluates under variable assignment $\theta$ to the three-valued set $\app{I}\theta(\{ \xxx \mid \varphi\}) := \{ \ddd^{\tval} \mid \app{I}\theta[\xxx/\ddd](\varphi) = \tval \}$. The \emph{minimal value} $\app{I}\theta(\agg(V))_{\min}$ of an aggregate term $\agg(V)$ in $\app{I}$ under $\theta$ is defined by 
\[ \app{I}\theta(\agg(V))_{\min} = \min\{ n \mid \text{$n = \agg(v)$ 
for some $v$ approximated by $\app{I}\theta(V)$} \}.  \]
Similarly, the \emph{maximal value} of $\agg(V)$ is defined by
\[ \app{I}\theta(\agg(V))_{\max} = \max\{ n \mid \text{$n = \agg(v)$ 
for some $v$ approximated by  $\app{I}\theta(V)$} \}. \]
The truth value of an \foagg formula $\varphi$ in structure
$\app{I}$ under variable assignment $\theta$ is defined by adding the
following cases to the definition of the truth value of an FO formula:
\begin{longitem}
	\item $
	\app{I}\theta(x \geq \agg(V)) = 
		\begin{cases} 
			\true		& \text{if $\theta(x)$ is a real number and $\theta(x) \geq \app{I}\theta(\agg(V))_{\max}$}\\
			\unkn		& \text{if $\theta(x)$ is a real number and $\theta(x) \geq \app{I}\theta(\agg(V))_{\min}$ and} \\
   & \theta(x) < \app{I}\theta(\agg(V))_{\max} \\ \false	& \text{otherwise.}
		\end{cases}$ 
	\item $
	\app{I}\theta(x \leq \agg(V)) = 
		\begin{cases} 
			\true		& \text{if $\theta(x)$ is a real number and $\theta(x) \leq \app{I}\theta(\agg(V))_{\min}$}\\
			\unkn		& \text{if $\theta(x)$ is a real number and $\theta(x) >
                        \app{I}\theta(\agg(V))_{\min}$ and}\\
&  \theta(x) \leq \app{I}\theta(\agg(V))_{\max} \\
			\false	& \text{otherwise.}
		\end{cases}$ 
\end{longitem}

\citeN{tplp/PelovDB07} illustrate that this definition of the value of \foagg formulas in three-valued structures is sufficiently precise for most applications found in the literature. They also show that the value of an \foagg formula in a three-valued structure $\app{I}$ can be computed in polynomial time in $\size{\app{I}}$.

\subsection{Propagation for \foagg}

To extend the propagation method to theories containing aggregates, the definition of INF sentences is extended to include aggregates. As in the case of FO, a propagator with polynomial-time data complexity is associated to each of these sentences. Next, it is shown that every \foagg theory over a vocabulary $\voc$ can be converted to a $\voc$-equivalent theory of INF sentences. To represent propagation on \foagg theories as a rule set and to allow symbolic propagation, the definition of $\ct{\varphi}$ and $\cf{\varphi}$ is extended to formulas $\varphi$ that may contain aggregates.

\begin{definition}
A \foagg sentence $\varphi$ is in \emph{implicational normal form} (INF) if it is of the form $\forall \xxx (\psi \lir L[\xxx])$, where $L[\xxx]$ is a literal that does not contain an aggregate and $\psi$ is a formula. The result of applying the INF propagator $\infp{\varphi}$ associated to $\varphi$ on a three-valued structure $\app{I}$ is defined as in Definition~\ref{def:infp}. If $\app{I}$ is strictly four-valued, then we define $\infp{\varphi}(\app{I}) = \mapp$.
\end{definition}
\begin{proposition}
For every INF sentence $\varphi$, $\infp{\varphi}$ is a monotone propagator with polynomial-time data complexity.
\end{proposition}
The proof of this proposition is analogous to the proof of Proposition~\ref{prop:infpisprop}.

We now show that every \foagg theory $T$ over vocabulary $\voc$ can be
converted to a $\voc$-equivalent theory containing only INF
sentences. Similarly as in the case of FO theories, we present a
conversion in several steps. None of the steps preserves complete propagation.
The following example indicates that even for very simple theories, complete polynomial-time propagation is impossible if $\pol \neq \np$.
\begin{example}
Let $T$ be the theory containing the sentence $\asum\{ x \mid P(x) \} = n$, where $n$ is a natural number. Let $\app{I}$ be a finite structure with domain $D \subset \N$ such that $P^{\app{I}}(d) = \unkn$ for every $d \in D$. Then $\rfi{T}(\app{I}) \neq \mapp$ iff $\sum_{d \in V} d = n$ for some subset $V \subseteq D$. Deciding whether such a subset exists is \np-complete \cite{Sipser05}. Hence if $\pol \neq \np$, $\rfi{T}$ cannot be implemented by a polynomial-time algorithm.
\end{example}

\begin{definition}
A \foagg sentence $\varphi$ is in \emph{equivalence normal form} (ENF) if $\varphi$ is an FO sentence in ENF or $\varphi$ is of the form $\forall \xxx\forall z\ (L[\xxx,z] \ler z \geq \agg\{ \yyy \mid L'[\xxx,\yyy] \})$ or $\forall \xxx\forall z\ (L[\xxx,z] \ler z \leq \agg\{ \yyy \mid L'[\xxx,\yyy] \})$.
\end{definition}
Every \foagg theory $T$ over $\voc$ can be rewritten to a $\voc$-equivalent theory in ENF by applying Algorithm~\ref{alg:fotoenf}.

\subsection{From \foagg to INF}

Similarly as for FO sentences in ENF, a set $\infs{\varphi}$ of INF sentences is associated to each ENF sentence $\varphi$ containing an aggregate. Our definition of this set is inspired by the propagation algorithms for propositional aggregates in the model generator \msid \cite{phd/Marien09}. These algorithms aim to restore bounds consistency. Intuitively, the propagators associated to the INF sentences we present, express that if some formula $y \geq \agg\{ \xxx \mid L[\xxx] \}$ must be true and the assumption that $L[\ddd]$ is true (respectively false) would imply that $y$ is strictly smaller than $\agg\{ \xxx \mid L[\xxx] \}$, then $L[\ddd]$ must be false (respectively true). Similarly for formulas of the form $y \leq \agg\{ \xxx \mid L[\xxx] \}$.

The following definition extends the definition of $\infs{\varphi}$ to the case where $\varphi$ is an ENF sentences containing an aggregate expression.
\begin{definition}
Let $\varphi$ be an ENF sentence of the form $\forall \xxx \forall z\ (H[\xxx,z] \ler z \geq \agg(\{ \yyy \mid L[\xxx,\yyy] \}))$.  Then $\infs{\varphi}$ is the set of INF sentences
\begin{align}
\label{eq:agginf1}	& \forall \xxx \forall z\ (z \geq \agg(V) \lir H[\xxx,z]), \\
\label{eq:agginf2}	& \forall \xxx \forall z\ (z < \agg(V) \lir \neg H[\xxx,z]), \\
\label{eq:agginf3}	& \forall \xxx \forall z \forall \yyy' \
(H[\xxx,z] \land z < \agg(\{\yyy \mid \yyy\neq \yyy' \land L[\xxx,\yyy]\}) \lir L[\xxx,\yyy']), \\
\label{eq:agginf4}	& \forall \xxx \forall z\forall \yyy' \
(H[\xxx,z] \land z < \agg(\{\yyy \mid  \yyy=\yyy' \lor L[\xxx,\yyy]\}) \lir \neg L[\xxx,\yyy']), \\
\label{eq:agginf5}	& \forall \xxx \forall z\forall \yyy' \ (\neg H[\xxx,z] \land z \geq \agg(\{\yyy \mid \yyy\neq \yyy' \land L[\xxx,\yyy]\}) \lir L[\xxx,\yyy']), \\
\label{eq:agginf6}	& \forall \xxx \forall z\forall \yyy' \ (\neg H[\xxx,z] \land z \geq \agg(\{\yyy \mid  \yyy=\yyy' \lor L[\xxx,\yyy]\}) \lir \neg L[\xxx,\yyy']).
\end{align}
For an ENF sentence $\varphi$ of the form $\forall \xxx \forall z\ (H[\xxx,z] \ler z \leq \agg(V))$, $\infs{\varphi}$ is defined similarly (it suffices to replace `$<$' by `$>$' and `$\geq$' by `$\leq$' in sentences~\eqref{eq:agginf1}--\eqref{eq:agginf6}).
\end{definition}

The INF sentences~\ref{eq:agginf3} and \ref{eq:agginf5} evaluate the aggregate on the set of tuples selected by the set expression but $\yyy'$ and express conditions for which the original ENF sentence $\varphi$ cannot be true unless $\yyy'$ is selected by the set expression.  Similarly, the INF sentences~\ref{eq:agginf4} and \ref{eq:agginf6} evaluate the aggregate on the set of the tuples selected by the set expression extended with $\yyy'$ and expresses conditions for which the original ENF sentence $\varphi$ cannot be true unless $\yyy'$ is not selected by the set expression.

Each of the sentences in $\infs{\varphi}$ is implied by $\varphi$. Vice versa, $\infs{\varphi}$ clearly implies $\varphi$ for every ENF sentence $\varphi$. Hence, we obtain the following proposition.

\begin{proposition}\label{prop:agginfs}
$\infs{\varphi}$ is equivalent to $\varphi$ for every ENF sentence $\varphi$.
\end{proposition}

\subsection{Rule Sets and Symbolic Propagation}\label{sec:aggrprop}

To represent propagation for FO as a rule set and to define symbolic propagation for FO theories, we relied on the fact that the value of an FO formula $\varphi$ in a three-valued structure can be found by computing the value of the negation-free formulas $\ct{\varphi}$ and $\cf{\varphi}$. Under certain conditions, it is possible to extend the definition of $\ct{\varphi}$ and $\cf{\varphi}$ to \foagg formulas $\varphi$. This immediately lifts the results of Section~\ref{ssec:itod} and Section~\ref{sec:sip} to \foagg.

It is beyond the scope of this paper to properly state the definition of $\ct{\varphi}$ and $\cf{\varphi}$ for an $\foagg$ formula $\varphi$, and the conditions under which this definition is the \emph{correct} one. We refer the reader to \cite[pages 90--91 and 181--184]{phd/Wittocx10} for these results. The results are correct for finite structures, embedded in infinite background structures. Relatively simple formulas $\ct{\varphi}$ and $\cf{\varphi}$ are obtained under the extra condition that all numbers that occur in structures are strictly positive and larger than $1$. If arbitrary real numbers are allowed, the formulas $\ct{\varphi}$ and $\cf{\varphi}$ become so complicated that they are not useful in practice. 

\section{Proofs}

\subsection*{Proof of Proposition~\ref{prop:fopredintro}}

Denote the vocabulary $\voc \cup \{ P \}$ by $\voc'$ and let $I$ be a $\voc$-structure. Any expansion of $I$ to $\voc'$ that satisfies the sentence $\forall \xxx\ (P(\xxx) \ler \psi[\xxx])$ necessarily assigns $\{ \xxx \mid \psi[\xxx]\}^I$ to $P$. Hence, such an expansion satisfies $\varphi'$ iff $I \models \varphi$. 

\subsection*{Proof of Proposition~\ref{prop:tfpos}}
Follows directly from the definition of $\ct{\varphi}$ and $\cf{\varphi}$.

\subsection*{Proof of Proposition~\ref{prop:cpandfo}}

Let $\app{I}$ and $\app{J}$ be two finite $\voc$-structure with the same domain. Denote their corresponding CSPs by $\langle \mathcal{C}_T,V_{\app{I}},\domf_{\app{I}} \rangle$, respectively $\langle \mathcal{C}_T,V_{\app{J}},\domf_{\app{J}} \rangle$. Then $V_{\app{I}} = V_{\app{J}}$. Also, $\app{I} \leqp \app{J}$ iff $\domf_{\app{I}} \supseteq \domf_{\app{J}}$. Therefore, $f$ is domain reducing iff $O(\app{I}) \geqp \app{I}$ for every structure $\app{I}$. 

Function $f$ is a propagator iff $\langle \mathcal{C}_T,V_{\app{I}},\domf_{\app{I}} \rangle$ and $\langle \mathcal{C}_T,V_{\app{I}},\domf_{\prp(\app{I})} \rangle$ have the same solutions. Because of the correspondence between models of $T$ approximated by $\app{I}$, respectively $\prp(\app{I})$, and solutions of $\langle \mathcal{C}_T,V_{\app{I}},\domf_{\app{I}} \rangle$, respectively $\langle \mathcal{C}_T,V_{\app{I}},\domf_{\prp(\app{I})} \rangle$, it follows that $f$ is a propagator iff the models of $T$ approximated by $\app{I}$ are precisely the models of $T$ approximated by $\prp(\app{I})$. 

We conclude that $\prp$ is a propagator for $T$ iff $f$ is a domain reducing propagator for CSPs of the form $\langle \mathcal{C}_T,V_{\app{I}},\domf_{\app{I}}\rangle$.

\subsection*{Proof of Lemma~\ref{lem:propcompose}}
Since $\prp_1$ and $\prp_2$ are propagators, $\app{I} \leqp \prp_2(\app{I}) \leqp \prp_1(\prp_2(\app{I})) = (\prp_1 \circ \prp_2)(\app{I})$ for every structure $\app{I}$. If $J \models T_1 \cup T_2$ and $\app{I} \leqp J$, then $\prp_2(\app{I}) \leqp J$ and therefore also $\prp_1(\prp_2(\app{I})) \leqp J$. Hence $\prp_1 \circ \prp_2$ is a propagator.

\subsection*{Proof of Proposition~\ref{prop:extendprop}}
Let $O'$ be a propagator for $T'$. Then for every $\voc$-structure $\app{I}$, $\app{I} = \res{(\app{I}+\lapp_{\voc'\setminus\voc})}{\voc} \leqp \res{(O'(\app{I}+\lapp_{\voc'\setminus\voc}))}{\voc} = O(\app{I})$. If $J$ is a model of $T$ such that $\app{I} \leqp J$, then there exists an expansion $J'$ of $J$ to $\voc'$ such that $J' \models T'$. Because $O'$ is a propagator, $O'(\app{I}+\lapp_{\voc'\setminus\voc}) \leqp J'$ and therefore $O(\app{I}) \leqp J$. We conclude that $O$ is a propagator. It is straightforward to check that if $O'$ is monotone, $O$ is also monotone.

\subsection*{Proof of Proposition~\ref{PROP:LEN}}\label{sec:lenp}

Recall that we defined $\size{\app{I}}$ as the cardinality of the domain of $\app{I}$. We prove that every sequence $\app{I} = \app{J}_0 <_p \app{J}_1 <_p \ldots <_p \app{J}_n$ has length polynomial in $\size{\app{I}}$. Denote by $N_P$ the number of predicate symbols in $\voc$. Let $A_P$ be the maximum arity of a predicate symbol in $\voc$. 

Since the sequence is increasing in precision, for every predicate symbol $P$ the number of $i$ such that $\ct{P}^{\tf{\app{J}_i}} \subsetneq \ct{P}^{\tf{\app{J}_{i+1}}}$ is at most $\size{\app{I}}^{A_P}$. Similarly, $\cf{P}^{\tf{\app{J}_i}}$ changes at most $\size{\app{I}}^{A_P}$ times in the sequence. Because $\langle \app{J}_i \rangle_{0 \leq i \leq n}$ is strictly increasing in precision, there is for every $0 \leq i < n$ at least one predicate $P$ such that $P^{\app{J}_i} \neq P^{\app{J}_{i+1}}$. Combining these results gives a maximum length of $2 \cdot \size{\app{I}}^{A_P} \cdot N_P$ for the sequence $\langle \app{J}_i \rangle_{0 \leq i \leq n}$. Clearly, this is polynomial in $\size{\app{I}}$. 

\subsection*{Proof of Proposition~\ref{prop:refs}}
Let $\langle \app{J}_{\xi} \rangle_{0 \leq \xi \leq \alpha}$ and $\langle \app{K}_{\xi} \rangle_{0 \leq \xi \leq \beta}$ be two stabilizing $V$-refinement sequences from $\app{I}$. Let $\langle \app{L}_{\xi} \rangle_{0 \leq \xi \leq \beta}$ the sequence of structures defined by 
\begin{longitem}
	\item $\app{L}_0 = \app{J}_{\alpha}$,
	\item $\app{L}_{\xi + 1} = \prp(\app{L}_{\xi})$ for every ordinal $0 \leq \xi < \alpha$, where $\prp$ is a propagator from $V$ such that $\app{K}_{\xi+1} = \prp(\app{K}_{\xi})$,
	\item $\app{L}_{\lambda} = \lub_{\leqp}(\{ \app{L}_{\xi} \mid 0 \leq \xi < \lambda \})$ for every limit ordinal $\lambda \leq \alpha$.
\end{longitem}
Because $\langle \app{J}_{\xi} \rangle_{0 \leq \xi \leq \alpha}$ is stabilizing, it follows that $\app{L}_{\beta} = \app{J}_{\alpha}$. Since $\app{I} \leqp \app{J}_{\alpha}$, we have that $\app{K}_{\beta} \leqp \app{L}_{\beta}$. Hence, we obtain that $\app{K}_{\beta} \leqp \app{J}_{\alpha}$. Similarly, we can derive that $\app{J}_{\alpha} \leqp \app{K}_{\beta}$. Hence $\app{J}_{\alpha} = \app{K}_{\beta}$. It follows that every stabilizing $V$-refinement sequence from $\app{I}$ has the same limit, namely $\app{J}_{\alpha}$.

\subsection*{Proof of Proposition~\ref{lem:complete}}
Follows immediately from the fact that $\{ M \mid \text{$\app{I} \leqp M$ and $M \models T$} \}$ is a superset of $\{ M \mid \text{$\app{J} \leqp M$ and $M \models T$} \}$ if $\app{I} \leqp \app{J}$.

\subsection*{Proof of Proposition~\ref{prop:opt}}

To prove the proposition, we show that $P^{\prp(\app{I})}(\ddd) \leqp P^{\rfi{T}(\app{I})}(\ddd)$ for any domain atom $P(\ddd)$. If $P^{\prp(\app{I})}(\ddd) = \inco$, it follows from the fact that $\prp$ is a propagator that there is no model of $T$ approximated by $\app{I}$. From the definition of $\rfi{T}$, we conclude that also $P^{\rfi{T}(\app{I})}(\ddd) = \inco$. If on the other hand $P^{\prp(\app{I})}(\ddd) = \true$ or $P^{\prp(\app{I})}(\ddd) = \false$, then $P(\ddd)$ is true, respectively false, in every model of $T$ approximated by $\app{I}$. Therefore $P^{\rfi{T}(\app{I})}(\ddd) \geq_p \true$, respectively $P^{\rfi{T}(\app{I})}(\ddd) \geq_p \false$, in this case. It follows that $P^{\prp(\app{I})}(\ddd) \leqp \rfi{T}(\app{I})(\ddd)$ for every domain atom of the form $P(\ddd)$. 

\subsection*{Proof of Proposition~\ref{prop:infpisprop}}\label{ssec:infpispropp}

Since $\varphi$ is an INF sentence, it is of the form $\forall \xxx\ (\psi \lir L[\xxx])$. Let $P$ be the predicate in $L[\xxx]$, i.e., $L[\xxx]$ is either the positive literal $P(\xxx)$ or the negative literal $\neg P(\xxx)$.

It follows directly from the definition of $\infp{\varphi}$ that $\app{I} \leqp \infp{\varphi}(\app{I})$ for every structure $\app{I}$. Now let $J$ be a structure such that $\app{I} \leqp J$ and $J \models \varphi$. To show that $\infp{\varphi}$ is a propagator, we have to prove that $P^{\infp{\varphi}(\app{I})}(\ddd_x) \leqp P^J(\ddd_x)$ for every tuple $\ddd_x$ of domain elements. If $\app{I}[\xxx/\ddd_x](\psi) \leqp \false$ then $P^{\infp{\varphi}(\app{I})}(\ddd_x) = P^{\app{I}}(\ddd_x) \leqp P^J(\ddd_x)$. If on the other hand $\app{I}[\xxx/\ddd_x](\psi) = \true$, then also $J[\xxx/\ddd_x](\psi) = \true$ and therefore $J[\xxx/\ddd_x](L[\xxx]) = \true$. It follows that $\app{I}[\xxx/\ddd_x](L[\xxx]) \leqp \true$ and hence $\infp{\varphi}(\app{I})[\xxx/\ddd_x](L[\xxx]) = \true$. We conclude that $P^{\infp{\varphi}(\app{I})}(\ddd_x) \leqp P^J(\ddd_x)$. 

The monotonicity of $\infp{\varphi}$ follows from the fact that $\app{I} \leqp \app{J}$ implies $\app{I}\theta(\psi) \leqp \app{J}\theta(\psi)$ for any two structures $\app{I}$ and $\app{J}$ and variable assignment $\theta$.

\subsection*{Proof of Proposition~\protect\ref{PROP:INPP}}\label{sec:inpp}

In the rest of this proof, let $\app{J}$ be the structure
$\lim_{\infps{V}}(\app{I})$. Observe that because $\app{I} \leqp
\app{J}$, $\tf{\app{I}} \leqt \tf{\app{J}}$. It now suffices to show
that $\tf{\app{J}}$ is a fixpoint of $\ico{\Delta_V}$ and that
$\tf{\app{J}} \leqt M$ holds for every fixpoint 
$M$
of $\ico{\Delta_V}$ such that $\tf{\app{I}} \leqt M$.

We first show that $\tf{\app{J}}$ is a fixpoint of $\ico{\Delta_V}$. Let $\forall \xxx (\psi \lir L[\xxx]) \in V$. Because $\app{J}$ is the limit of a $\infps{V}$-refinement sequence, $\app{J}\theta(\psi) \leqp \app{J}\theta(L[\xxx])$ for every variable assignment $\theta$. Hence, if $\tf{\app{J}}\theta(\ct{\psi}) = \true$, then $\tf{\app{J}}\theta(\ct{\psi}) = \true$. It follows that $\ico{\Delta_V}(\tf{\app{J}}) = \tf{\app{J}}$, i.e., $\tf{\app{J}}$ is a fixpoint of $\ico{\Delta_V}$.

To show that $\app{J}$ is more precise than all other fixpoints of $\ico{\Delta_V}$ that are more precise than $\app{I}$, let $\langle \app{K}_{\xi} \rangle_{0 \leq \xi \leq \alpha}$ be a stabilizing $\infps{V}$-refinement sequence from $\app{I}$. Then $\app{K}_{\alpha} = \app{J}$. Let $\langle L_{\xi} \rangle_{0 \leq \xi \leq \alpha}$ be the sequence of $\tf{\voc}$ structures defined by $L_0 = \tf{\app{I}}$, $L_{\xi + 1} = \ico{\Delta_V}(\xi)$ for every $\xi < \alpha$, and $L_{\lambda} = \lub_{\leqt}(\{ L_{\xi} \mid \xi < \lambda\})$ for every limit ordinal $\lambda \leq \alpha$. Because $\ico{\Delta_V}$ is $\leqt$-monotone, it follows from Tarski's theorem that $L_{\alpha} \leqt \glb_{\leqt}\{ M \mid M \models \Delta_V \text{ and } M \geqt \tf{\app{I}}\}$. 
Since the propagators used in the stabilizing refinement sequence
$\langle \app{K}_{\xi} \rangle_{0 \leq \xi \leq \alpha}$ are part of
the rule set $\Delta_V$, it is straightforward to check that
$\tf{\app{K}_{\xi}} \leqt L_{\xi}$ for every $0 \leq \xi \leq \alpha$;
hence it follows that $\tf{\app{J}} \leqt \glb_{\leqt}\{ M \mid M \models \Delta_V \text{ and } M \geqt \tf{\app{I}}\}$.

\subsection*{Proof of Proposition~\ref{lem:ev}}

If $\varphi$ is the atomic formula $P(\yyy)$ and $P^{\Phi} = \{ \xxx
\mid \psi \}$, then $\Phi(E)\theta \models P(\yyy)$ iff $\theta(\yyy)
\in P^{\Phi(E)}$ iff $\theta(\yyy) \in \{ \xxx \mid \psi \}^E$ iff
$E\theta[\xxx/\theta(\yyy)] \models \psi$ iff $E\theta \models
\psi[\xxx/\yyy]$ iff $E\theta \models \Phi(P(\yyy))$. The cases were
$\varphi$ is not atomic easily follow by induction. 

\subsection*{Proof of Proposition~\ref{prop:sinfp}}\label{ssec:sinfpp}

We prove the case where $\varphi$ is of the form $\forall \xxx\ (\psi \lir L[\xxx])$ and $L[\xxx]$ is a positive literal. The proof is similar in case $L[\xxx]$ is a negative literal. 

Let $E$ be a $\symvoc$-structure and $\app{\Phi}$ a four-valued
symbolic $\voc$-structure over $\symvoc$. Let $\varphi$ be the INF
sentence $\forall \xxx\ (\psi \lir P(\xxx))$. We have to show that
$\sinfp{\varphi}(\app{\Phi})(E) =
\infp{\varphi}(\app{\Phi}(E))$. Therefore, we must prove that
$Q^{\tf{\sinfp{\varphi}(\app{\Phi})(E)}} =
Q^{\tf{\infp{\varphi}(\app{\Phi}(E))}}$ for every predicate $Q \in
\tf{\voc}$.

First assume $Q \neq \ct{P}$. Then the following is a correct chain of equations.
\begin{equation}\label{eq:eqchain}
Q^{\tf{\infp{\varphi}(\app{\Phi}(E))}} = Q^{\tf{\app{\Phi}(E)}} 
											 = (Q^{\app{\Phi}})^E
											 = (Q^{\sinfp{\varphi}(\app{\Phi})})^E
											 = Q^{\tf{\sinfp{\varphi}(\app{\Phi})(E)}}
\end{equation}
The first and third
equality follow from the definitions of $\infp{\varphi}$,
respectively $\sinfp{\varphi}$, and the assumption that $Q \neq
\ct{P}$. The second and the fourth 
equality apply the definition of $\app{\Phi}(E)$. 

The following chain shows that also
$\ct{P}^{\tf{\sinfp{\varphi}(\app{\Phi})(E)}} =
\ct{P}^{\tf{\infp{\varphi}(\app{\Phi}(E))}}$: 
\begin{align*}
\ct{P}^{\tf{\sinfp{\varphi}(\app{\Phi})(E)}} 
	& = \left(\ct{P}^{\sinfp{\varphi}(\app{\Phi})}\right)^E \\
	& = \left(\ct{P}^{\app{\Phi}} \cup \{ \xxx \mid \app{\Phi}(\ct{\psi}) \} \right)^E \\
	& = \left(\ct{P}^{\app{\Phi}}\right)^E \cup \left(\{ \xxx \mid \app{\Phi}(\ct{\psi}) \} \right)^E \\
	& = \ct{P}^{\tf{\app{\Phi}(E)}} \cup \{ \ddd \mid E[\xxx/\ddd] \models \app{\Phi}(\ct{\psi}) \} \\
	& = \ct{P}^{\tf{\app{\Phi}(E)}} \cup \{ \ddd \mid (\app{\Phi}(E))[\xxx/\ddd])(\psi) \geqp \true \} \\
	& = \ct{P}^{\tf{\infp{\varphi}(\app{\Phi}(E))}}
\end{align*}

The first equality 
follows from the definition of $\app{\Phi}(E)$, the second one from the definition of $\sinfp{\varphi}$, the third one from the definition of union of queries, the fourth one from the definition of $\app{\Phi}(E)$ and of query evaluation, the fifth one from Proposition~\ref{prop:ftot}, and the final one from the definition of $\infp{\varphi}$.

\subsection*{Proof of Proposition~\ref{prop:agginfs}}

We prove the case where $\varphi$ is of the form $\forall \xxx \forall z\ (P(\xxx,z) \ler z \geq \agg(\{ \yyy \mid L[\xxx,\yyy] \}))$. Clearly, $\varphi$ is equivalent to the conjunction of~\eqref{eq:agginf1} and~\eqref{eq:agginf2}. Hence we only have to show that~\eqref{eq:agginf3}--\eqref{eq:agginf6} are implied by $\varphi$. 

We show that~\eqref{eq:agginf3} is implied by $\varphi$. The proofs
that~\eqref{eq:agginf4}--\eqref{eq:agginf6} are implied by $\varphi$
are similar. Let $I$ be a structure and $\theta$ a variable assignment
such that $I \models \varphi$ and $I\theta \models P(\xxx,z) \land z <
\agg(\yyy \mid \yyy\neq \yyy' \land L[\xxx,\yyy])$. Since $I \models \varphi$ and $I\theta
\models P(\xxx,z)$, $I\theta \models z \geq \agg(\{ \yyy \mid L[\xxx,\yyy] \})$. Because $I\theta
\models z < \agg(\yyy \mid \yyy\neq \yyy' \land L[\xxx,\yyy])$, it follows that $I\theta(\{ \yyy \mid L[\xxx,\yyy] \}) \neq
I\theta(\yyy \mid \yyy\neq \yyy' \land
  L[\xxx,\yyy])$. Hence $I\theta \models L[\xxx/\yyy']$. We conclude that $I
\models \eqref{eq:agginf3}$.

The case where $\varphi$ is of the form $\forall \xxx \forall z\ (P(\xxx,z) \ler z \leq \agg(\{ \yyy \mid L[\xxx,\yyy] \}))$ is analogous.

\begin{acks}
  This research was supported by projects G.0357.06 and G.0489.10 of
  Research Foundation - Flanders (FWO-Vlaanderen) and by GOA/08/008
  "Probabilistic Logic Learning".
\end{acks}

\bibliographystyle{acmtrans}
\bibliography{krrlib}
\end{document}